\documentclass[conference,letterpaper]{IEEEtran}

\usepackage[utf8]{inputenc} 
\usepackage[T1]{fontenc}
\usepackage{url}
\usepackage{ifthen}
\usepackage{cite}
\usepackage[cmex10]{amsmath} 
\usepackage{mathrsfs}

\usepackage[colorlinks,linkcolor=blue,citecolor=blue,urlcolor=blue]{hyperref}

\usepackage{amsmath,amssymb,graphicx,stmaryrd,cleveref,latexsym,fullpage}
\usepackage{amsopn,amsthm,upref}

\usepackage{cite}
\usepackage{url}
\usepackage[caption=false]{subfig}

\usepackage{comment}

\usepackage{tikz}
\usetikzlibrary{cd}

\newcommand{\Iff}{\textbf{if\textcompwordmark f} }
\newcommand{\Tau}{\mathrm{T}}
\newcommand{\assign}{=}

\newcommand{\tmem}[1]{{\em #1\/}}

\newcommand{\tmop}[1]{\ensuremath{\operatorname{#1}}}

\newcommand{\rbr}[1]{\left(#1\right)}

\newcommand{\dtp}[2]{{\langle #1, #2\rangle}}

\let\le\leqslant
\let\ge\geqslant
\let\phi\varphi

\newcommand{\eps}{\varepsilon}

\newcommand{\vv}{\boldsymbol{v}}

\newcommand{\va}{\boldsymbol{a}}

\newcommand{\bz}{\mathbf{0}}
\newcommand{\ch}{\operatorname{char}}

\newcommand{\h}{\eta}

\theoremstyle{plain}
\newtheorem{lemma}{Lemma}

\newtheorem{theorem}{Theorem}

\newtheorem{problem}{Problem}

\theoremstyle{definition}

\newtheorem{definition}{Definition}

\theoremstyle{remark}
\newtheorem{remark}{Remark}

\newtheorem*{remark*}{Remark}

\newcommand{\RR}{\mathbb{R}}

\newcommand{\PP}{\mathbb{P}}

\newcommand{\abs}[1]{\left\lvert #1 \right\rvert}
\newcommand{\norm}[1]{\left\lVert#1\right\rVert}

\newcommand{\T}{\mathrm{T}}
\newcommand{\id}{\mathrm{id}}
\newcommand{\R}{\ensuremath{\mathbb{R}}}

\newcommand{\cA}{\mathcal{A}}
\newcommand{\cB}{\mathcal{B}}
\newcommand{\cC}{\mathcal{C}}

\newcommand{\cQ}{\mathcal{Q}}

\newcommand{\cL}{\mathcal{L}}

\newcommand{\cV}{\mathcal{V}}
\newcommand{\cF}{\mathcal{F}}
\newcommand{\brF}{\bar{\mathcal{F}}}

\newcommand{\brC}{\bar{\mathcal{C}}}

\newcommand{\LP}{\mathrm{LP}}

\newcommand{\LC}{\mathbf{LC}}

\newcommand{\Ab}{\mathbf{Ab}}
\newcommand{\Rng}{\mathbf{Ring}}
\newcommand{\Vc}{\mathbf{Vect}}

\newcommand{\bC}{\mathbf{C}}
\newcommand{\bZ}{\mathbf{Z}}
\newcommand{\bB}{\mathbf{B}}

\newcommand{\cay}{\mathrm{Cay}}

\newcommand{\fF}{\mathfrak{F}}

\newcommand{\Sym}{\mathbf{S}}
\newcommand{\Z}{\mathbb{Z}}

\newcommand{\F}{\mathbb{F}}

\DeclareMathOperator{\sk}{\mathbf{sk}}
\DeclareMathOperator{\rk}{rk}
\DeclareMathOperator{\im}{im}

\DeclareMathAlphabet\mathbfcal{OMS}{cmsy}{b}{n}
\newcommand{\bcF}{\mathbfcal{F}}

\newcommand{\bcL}{\mathbfcal{L}}

\begin{document}
\title{Maximally Extendable Sheaf Codes} 

\author{%
  \IEEEauthorblockN{Pavel Panteleev and Gleb Kalachev}
  \IEEEauthorblockA{Department of Mechanics and Mathematics\\
                    Moscow State University\\
                    Moscow, Russia\\
                    Email: \{panteleev,  kalachev\}@intsys.msu.ru}
}


\maketitle


\begin{abstract}

We study sheaf codes, a type of linear codes with a fixed hierarchical collection of local codes, viewed as a sheaf of vector spaces on a finite topological space we call coded space. Many existing codes, such as tensor product codes, Sipser-Spielman codes, and their more recent high-dimensional analogs, can be naturally represented as sheaf codes on simplicial and cubical complexes, considered as coded spaces. We introduce a new property of a sheaf code, called maximal extendibility, which ensures that within a class of codes on the same coded space, we encounter as few obstructions as possible when extending local sections globally. We show that in every class of sheaf codes defined on the same space and parameterized by parity-check matrices with polynomial entries,  there always exists a maximally extendable sheaf code. Such codes are very interesting since it is possible to show that maximally extendable tensor product codes are good coboundary expanders, which potentially could be used to attack the qLTC conjecture.
\end{abstract}

\section{Introduction}

The recent constructions of asymptotically good\footnote{An infinite family of classical or quantum codes is called (\emph{asymptotically}) \emph{good} if both the dimension and the distance of the codes grows as $\Theta(n)$ with the length $n\to\infty$.} classical locally testable codes (LTCs) and quantum low-density parity-check (qLDPC) codes~\cite{Dinur:2021,Panteleev&Kalachev:2021:ltc, leverrierQuantumTannerCodes2022a, dinurGoodQuantumLDPC2023} are all based on the idea of lifting a small local (tensor) product code to a~large global code. However, for these constructions to work, it is required that the local product code has a~very specific property that can be expressed topologically as a~coboundary expansion of a~bipartite graph with a~sheaf on it~\cite{kalachevTwosidedRobustlyTestable2023}. When we lift the product code, its local coboundary expansion propagates to the~global one\footnote{Originally this was only formally shown for the \emph{quantum Tanner codes}~\cite{leverrierQuantumTannerCodes2022a}, the simplified version of the \emph{expander lifted product codes}~\cite{Panteleev&Kalachev:2021:ltc}. However, this holds for \emph{all} current constructions~\cite{golowichNLTSHamiltoniansStronglyExplicit2023}. }~\cite{hopkinsExplicitLowerBounds2022a}, which implies many interesting properties of the obtained codes, including single-shot~\cite{guSingleshotDecodingGood2023} linear time decoding~\cite{Gu:stoc2023:qpdpc-decoder, Leverrier:qldpcdecoder:2023, Dinur:2021}, NLTS property~\cite{anshuNLTSHamiltoniansGood2023}, \cite{golowichNLTSHamiltoniansStronglyExplicit2023}, and linear energy barriers~\cite{williamsonLayerCodes2023}.  
Good\! qLDPC\! codes\! can\! also\! be\! converted to\! geometrically local codes~\cite{portnoyLocalQuantumCodes2023},\! \cite{williamsonLayerCodes2023, linGeometricallyLocalQuantum2023}.  

Sheaf theoretic approach to classical linear codes, has spurred a lot of interest recently~\cite{dinurHighDimensionalExpanders2017, Meshulam:2018, Breuckmann:balanced:2021}, \cite{Panteleev&Kalachev:2021:ltc, kalachevTwosidedRobustlyTestable2023, firstGoodQueryLocally2023, dinurNewCodesHigh2023, dinurqLTC:2024}. In fact, already in~\cite{dinurHighDimensionalExpanders2017}, Irit Dinur and Tali Kaufman study\footnote{In the language of algebraic topology,  what they study corresponds to the idea of \v{C}ech cohomology of a~sheaf.}, in the context of property testing, the \emph{obstructions} to gluing together \emph{local} pieces of data into the \emph{global} one, and introduced \emph{agreement tests}.

In the current work, we conceptualize and generalize these ideas along with more recent research~\cite{firstGoodQueryLocally2023, dinurNewCodesHigh2023, dinurqLTC:2024}, and introduce a~new concept that we call a~\emph{sheaf code}. Essentially, a~sheaf code is just a~specific type of Tanner code 
$\cF(X) = \{c\in\F_2^n \mid c|_{X_\sigma} \in \cF_\sigma, \sigma\in X \}$
defined by a \emph{hierarchical collection}\footnote{It means that for $\sigma\le \tau$ in $X$ we have $X_\tau \subseteq X_\sigma$, and $c\in\cF_\sigma$ implies $c|_{X_\tau}\in\cF_\tau$ . We~also assume that $X$ has \emph{exactly} $n$ maximal elements identified with the elements of $[n]$, and $\cF_{\sigma} = \F_2$, $\sigma\in [n]$.} of local codes $\cF_\sigma\subseteq\F_2^{X_\sigma}$, $X_\sigma\subseteq [n]$, indexed by some poset $X$, playing the role of a~topological space\footnote{In fact, finite topological spaces and finite posets are almost the same thing (see, Appendices~\ref{sc:gen-sheaf} and \ref{sc:posets}).}, and called its \emph{coded space}\footnote{This term is inspired by the term ``ringed space'' from  sheaf theory, meaning a~topological space equipped with a~sheaf of rings on it.}. This allows one to deal with such linear codes using topological methods. In particular, it is possible to study cohomology groups $H^i(X; \cF)$ of these codes with local coefficients in $\cF$ (see Section~\ref{sc:cohom}).

The groups $H^i(X; \cF)$ are very interesting since they give us a~very general way to define quantum codes out of the classical sheaf codes~\cite{firstGoodQueryLocally2023, dinurNewCodesHigh2023, dinurqLTC:2024}. In fact, many existing qLDPC codes can be described in this way. For example, the famous Kitaev toric code~\cite{Kitaev:2002} defined on the natural cellulation of a~torus $X = C_{\ell}\times C_{\ell}$, viewed as a cell poset, where $C_\ell$ is a cycle graph with $\ell$ vertices, is $H^1(X; \underline{\F_2})$ if we denote by $\underline{\F_2}(X)$ the~\emph{constant sheaf code} on~$X$, where all local codes are repetition codes. 

Another example is the codes form~\cite{dinurGoodQuantumLDPC2023}, obtained by changing the roles of checks and qubits in the expander lifted product codes\footnote{While both codes are lifted product codes~\cite{Panteleev&Kalachev:2021}, they also belong to the classes of fiber bundle~\cite{Hastings:2021:fiber} and balanced product~\cite{Breuckmann:balanced:2021} codes.}~\cite{Panteleev&Kalachev:2021:ltc}. These codes can be represented as $H^1(X;\cF)$, where $\cF(X) = \cA(\Gamma)\otimes_G \cB(\Gamma)$ is the~balanced product (Definition~\ref{df:balanced-prod-sheaf}) of two Z\'{e}mor codes~\cite{Zemor:2001} $\cA(\Gamma)$ and $\cB(\Gamma)$, considered as sheaf codes on~$\Gamma$. Note that the classical codes $\cF(X)$ are essentially the same as the~good classical LTCs constructed in~\cite{dinurLocallyTestableCodes2022}.

A sheaf theoretic approach to Tanner codes allows us to define a~new property of such codes, that we call \emph{maximal extendibility}. This idea is inspired by the notion of \emph{maximally recoverable} (\emph{MR}) codes~\cite{chenMaximallyRecoverableProperty2007, gopalanExplicitMaximallyRecoverable2014}, studied in the context of locally recoverable codes~\cite{tamoFamilyOptimalLocally2014, golowichQuantumLocallyRecoverable2023}.  

Given a sheaf code $\cF(X)\subseteq \F^n$,  we say that an index set $U\subseteq X$ of local codes is \emph{extendable} if every ``local'' codeword ${c\in\F^{\cup_{\sigma\in U}X_\sigma}}$ satisfying only local  constraints $c|_{X_\sigma}\in\cF_\sigma$, ${\sigma\in U}$, can be extended to a~``global'' codeword ${\hat{c}\in \cC}$, satisfying all the constraints\footnote{In the language of sheaf theory, ``local'' and ``global'' codewords are called respectively \emph{local} and \emph{global sections}.}. Furthermore, we say that $\cC$ is \emph{maximally extendable} for a~class of codes on the space $X$ (e.g., all product codes $\cA\otimes \cB$ with fixed dimensions of $\cA, \cB\subseteq \F_2^n$) if every set $U\subseteq X$ extendable in \emph{some} code from this class is also extendable in $\cC$.  As it turns out, the maximally extendable\footnote{Since every information set $U\subseteq [n]\subseteq X$ for $\cF(X)$ is clearly extendable, then for a~maximally extendable code in a~family of codes of the same dimension, if $S$ is an~information set in \emph{some} code from this family, then $S$ is also an~information set of $\cF(X)$. Hence, every maximally extendable code is also maximally recoverable. However, we do not know whether the converse is also true.} product codes have asymptotically optimal coboundary expansion properties~\cite{kalachev:2024} for \emph{any} number of component codes. Previously this was only known in the case of \emph{two} codes~\cite{dinurGoodQuantumLDPC2023,kalachevTwosidedRobustlyTestable2023}.

The expansion properties of product codes in more than two dimensions are very important in the context of quantum LTCs~\cite{Aharonov:2015}. One possible way to attack the qLTC conjecture~\cite{eldarLocalHamiltoniansWhose2017}, \cite{crossQuantumLocallyTestable2023}, \cite{dinurqLTC:2024}, positing the existence of good locally testable qLDPC codes (qLTCs), is to move away from the current paradigm~\cite{Panteleev&Kalachev:2021:ltc, leverrierQuantumTannerCodes2022a, dinurGoodQuantumLDPC2023, dinurqLTC:2024} of defining codes on $\ell$-fold covering of products of graphs, to a~more general paradigm of defining codes on arbitrary cubical complexes, which only \emph{locally} looks like graph products. We think that an~interesting direction of future research is to define sheaf codes on Ramanujan cubical complexes constructed in~\cite{rungtanapiromInfiniteSeriesQuaternionic2019}. 

The rest of the paper is organized as follows. In Section~\ref{sc:def}, we introduce main definitions and notations. In section~\ref{sc:sheaf}, we define sheaf codes and prove our main result. In sections \ref{sc:examples}--\ref{sc:exp-sheaf-codes} we give examples, introduce operations, and, finally, study the cohomology and the expansion of sheaf codes. We also have several appendices with some supplementary information.

\section{Basic Definitions and Notations}\label{sc:def}

Given a~set $S$ and a~field $\F$, let $\F^S$ be the vector space of all formal $\F$-linear combinations $v = \sum_{s\in S} v(s)\cdot s$, which elements we also represent as vectors $(v(s))_{s\in S}$ over $\F$ indexed by the set~$S$. Thus, $S \subseteq \F^S$ is a basis of the vector space $\F^S$, and $\F^S \le \F^T$ when $S\subseteq T$. 

We view matrices $M\in \F^{A\times B}$ as tables $M(a,b)$ indexed by $a\in A$, $b\in B$, and denote by $M|_{J}$ the restriction $(M(a,b))_{(a,b)\in A\times J}$ of $M$ to a~subset of column indices $J\subseteq B$.  We also say that a~matrix $M$ is \emph{$w$-limited} if the number of non-zero elements in every row and column of $M$ is no more than $w$. 

A \emph{poset} is a~set $X$ equipped with a~partial order~$\le$. A~\emph{chain} in $X$ is a~subset where all elements are comparable. The \emph{height} of a~poset $X$ is the size of its largest chain minus $1$. A~poset of finite height can be described by a~special graph, called its \emph{Hasse diagram}, where the set of vertices is $X$, and we connect $\sigma,  \tau\in X$ by an~edge and write ${\sigma\prec \tau}$ \Iff ${\sigma < \tau}$ and there is no $\pi\in X$ such that $\sigma < \pi <\tau$. Furthermore, we often assume that a~poset $X$ comes with a~\emph{grading map} $\rho\colon X\to\Z$  such that ${\sigma\prec \tau}$ implies $\rho(\tau)=\rho(\sigma)+1$, in which case we call it a~\emph{graded poset}. We consider graphs, simplicial, and cubical complexes as graded poset, where the grading is the dimension $\dim \sigma = i$ of a~cell $\sigma$ called an~\emph{$i$-cell}. We put $X(i) = \{\sigma\in X\mid \dim \sigma = i\}$ and call the subposet $\sk_i X = \cup_{j\le i} X(j)$ the \emph{$i$-skeleton} of $X$.

We write $\cV \le \cV'$ if a~vector space $\cV$ is a~subspace of a vector space~$\cV'$. 
A~linear $[n,k,d]$ code over $\F$ is a~subspace $\cC\le \F^S$, $\abs{S} = n$, such that $k=\dim \cC$,  $d=\min_{c\in\cC\setminus\{0\}} \abs{c}$, where $\abs{c}$ is the \emph{Hamming weight} of the vector $c$, i.e., the number of non-zero elements $c(x)$, ${x\in S}$. The parameters $n$, $k$, and $d$ are called respectively the \emph{length}, the \emph{dimension}, and the \emph{distance} of $\cC$. The code is usually defined either as $\cC = \ker H$ or as $\cC = \im G^\T$, in which case $H$ and $G$ are called respectively its \emph{parity-check} and \emph{generator} matrices. We~also say that $I\subseteq S$ is an~\emph{information set} for $\cC$ if $\abs{I} = \dim \cC = \dim \cC|_I$, where $\cC|_I = \{c|_I \mid c\in C\}$. 
The~\emph{dual code} for $\cC\le \F^S$ is the code 
\[\cC^\perp = \{a\in\F^S\mid \dtp{a}{b} = 0 \text{ for all } b\in\cC\},\] 
where $\dtp{a}{b} = \sum_{x\in S} a(x) b(x)$ is the \emph{dot product} in $\F^S$.

By a~quantum $\llbracket n, k, d \rrbracket$~code over $\F$, we mean a~CSS code~\cite{CSS:1996, CSS2:1996}, which we represent as a~quotient $Q = \cA/\cB$ of two linear codes $\cA,\cB\le \F^{S}$, and define its length as $n = \abs{S}$, its \emph{dimension} as $k = \dim \cA - \dim \cB$, and its \emph{minimal distance} as $d = \min(d_X, d_Z)$, where $d_X = \min_{c\in \cA\setminus\cB}\abs{c}$ and $d_Z = \min_{c\in \cB^\perp\setminus\cA^\perp}\abs{c}$. Sometimes we omit $d$ and just say $[n,k]$ and $\llbracket n, k\rrbracket$ code.

\section{Sheaf Codes}\label{sc:sheaf}

\subsection{Sheaves in Topology}\label{ssc:topology}

Informally, a~sheaf is a~map $\cF$ assigning to \mbox{every} open set $U$ of a~topological space $X$ some set $\cF(U)$ of elements interpreted as ``functions on $U$'' such that for every ``compatible'' family of \emph{local} functions ${s_i \in \cF(U_i)}$, ${i\in I}$, there exists a~unique \emph{global} function $s \in \cF(\cup_{i\in I} U_i)$ ``compatible'' with all of them. 

A~canonical example is the~sheaf of continuous real functions on a~topological space $X$, assigning to each open set $U\subseteq X$ the set $\cF(U)$ of all continuous functions $s\colon U \to \RR$. It is not hard to see that if a~collection of local continuous functions $s_i\colon U_i \to \RR$, $i\in I$, is compatible (i.e., they pairwise agree: $s_i|_{U_i \cap U_j} = s_j|_{U_i \cap U_j}$ for all $i, j\in I$), we can glue them together into a~unique global continuous function $s\colon \cup_{i\in I} U_i \to \RR$ such that $s|_{U_i} = s_i$ for all $i\in I$. Another well-known example is a~\emph{constant sheaf} $\underline{\F}$ defined for arbitrary set $\F$ (e.g., the field $\R$), where $\cF(U)$ is the set of all constant functions $f\colon U\to \F$.

Formally, given a~topological space $X$, a \emph{sheaf} on $X$ is a function $\cF$ assigning to each open set $U\subseteq X$ the set $\cF(U)$ of elements called \emph{local sections} and to each pair of nested open sets $V\subseteq U$ a~map ${\cF_{U\to V}\colon \cF(U) \to \cF(V)}$ called a~\emph{restriction map} such that:
\begin{enumerate}
    \item[(1)] for all open sets $W \subseteq V \subseteq U$, we have: 
    \[\cF_{V \to W} \circ \cF_{U \to V} = \cF_{U \to W} \text{ and } \cF_{U\to U} = \id_U;\]
    \item[(2)] if a~collection $(s_i)_{i\in I}$ of local sections $s_i\in \cF(U_i)$ on open sets $U_i\subseteq X$ is \emph{compatible} (i.e., they pairwise agree: $\cF_{U_i\to U_i \cap U_j} (s_i) = \cF_{U_j\to U_i \cap U_j} (s_j)$ for all $i, j\in I$), then there exists a~unique \emph{global} section $s\in \cF(U)$ on $U = \cup_{i\in I} U_i$ compatible with all of them: $\cF_{U\to U_i}(s)=s_i$, $i\in I$. 
\end{enumerate}

Very often the sets $\cF(U)$ come with some additional structure, making them abelian groups, rings, or vector spaces over a field $\F$, in which case we say that we have a~\emph{sheaf of abelian groups}, \emph{rings}, or \emph{vector spaces}. Here, it is additionally required that all restriction maps are morphisms in the corresponding categories $\Ab, \Rng, \Vc_\F$. For example, since in a~sheaf $\cF$ of continuous real functions we can add, multiply, and make $\R$-linear combinations of the elements from $\cF(U)$, we can also view $\cF$ as a~sheaf with the values in the categories $\Ab, \Rng, \Vc_\R$, respectively; and one can easily check that the restriction maps $\cF_{U\to V}\colon s \mapsto s|_V$ here are indeed morphisms in these categories.

Since every code $\cC\le \F^S$ is just a~set of functions $c\colon S\to \F$, it is natural to extend the idea of sheaves to codes. It was observed by Assmus~~\cite{assmusCategoryLinearCodes1998} that the~class of linear codes over a~field $\F$ can be viewed as a~category, where the objects are codes and the morphisms are $\F$-linear maps $\phi\colon \cA\to \cB$ such that $\abs{\phi(x)} \le \abs{x}$ for all $x\in\cA$. However, for our purposes, it is much more convenient to consider the category $\LC_\F$ of linear codes over $\F$, where a~\emph{morphism} between codes $\cA \le \F^A$ and $\cB \le \F^B$ is an~$\F$-linear map $\phi\colon \F^A\to \F^B$ such that $\phi(\cA) \subseteq \cB$ and the matrix of $\phi$ is $1$-limited. Note that the notion of \emph{isomorphism} in $\LC_\F$ (i.e., a~bijective morphism $\phi$ such that $\phi^{-1}$ is also a~morphism) is exactly the same as in the category of Assmus, and coincides\footnote{Indeed, the classical McWilliam's extension theorem~\cite{Macwilliams:1962} asserts that any partial $\F$-linear isometry $\phi\colon \cA\to\cB$ of codes $\cA\le \F^{A}$ and $\cB\le \F^{B}$ can be extended to a~full $\F$-linear isometry $\hat{\phi}\colon \F^{A}\to\F^{B}$, where by \emph{isometry} we mean a~map preserving Hamming distance.} with the notion of \emph{monomial equivalence} of codes.  Moreover, in the proposed category $\LC_\F$ the \emph{dual} $\phi^*\colon \F^B \to \F^A$ of a~morphism $\phi\colon \F^A\to \F^B$, given by the transpose matrix, is also a~morphism, which is not the case for the category of Assmus.
Hence, we obtain the following general definition\footnote{The general idea to assign linear codes to open sets of topological space was originally presented by the first author in~\cite{Panteleev:2022:qmath}. It is just a~natural reformulation and generalization of the ideas proposed by Irit Dinur and Tali Kaufman in~\cite{dinurHighDimensionalExpanders2017}. Later, in the course of writing this manuscript, we became aware of other relevant works~\cite{firstGoodQueryLocally2023, dinurNewCodesHigh2023,dinurqLTC:2024}.}.
\begin{definition}[sheaf of linear codes]
    A \emph{sheaf of linear codes} over a~field $\F$ on a~topological space $X$ is a~sheaf of vector spaces over $\F$ on $X$, where the local sections are linear codes and the restriction maps are morphisms from the category $\LC_\F$. 
\end{definition}

Now we are ready to give the main definition.
\begin{definition}[sheaf code, general version]\label{df:sheaf-code-gen}
    A \emph{sheaf code} is a~linear code $\cC$ equipped with a sheaf $\cF$ of linear codes on a~topological space $X$, which we call its \emph{coded space}, such that $\cC = \cF(X)$, i.e., $\cC$ is the set of \emph{global sections}. We often omit $X$ and denote a~sheaf code $\cF(X)$ as $\cF$.
\end{definition}

\subsection{Sheaves on Posets and Tanner Sheaves}\label{sc:tanner-sheaf}

It~is known that every poset $X$ can be considered as a~topological space with the \emph{Alexandrov topology}, where the open sets are \emph{upper sets}, i.e., the sets $U\subseteq X$ such that $x\in U$ and $x\le y$ always implies $y\in U$. Moreover, \emph{finite} topological spaces are essentially the same thing as posets (see Appendices~\ref{sc:gen-sheaf} and \ref{sc:posets}). Therefore, in what follows, instead of arbitrary topological spaces we restrict ourselves to posets.  Given a~poset $X$, let $X_*$ be the set of all its maximal elements, $X_\sigma = \{x\in X_*\mid x\ge \sigma \}$, $X_{\ge \sigma } = \{x \in X\mid x \ge \sigma \}$, and $X_U = \cup_{\sigma\in U} X_\sigma$.

\begin{definition}[sheaf code]\label{df:tanner_sheaf}
    A~\emph{sheaf code} is a~code\footnote{Such codes are usually called \emph{Tanner} or \emph{generalized LDPC codes}.} 
    \[\cF(X) = \{c\in \F^{X_*} \mid  \forall \sigma\in X: c|_{X_\sigma}\in \cF_\sigma \}\]
    together with a~fixed hierarchical collection $(\cF_\sigma)_{\sigma\in X}$ of \emph{local codes} $\cF_\sigma\subseteq \F^{X\sigma}$ indexed by the elements of a~graded~poset $X$ called its \emph{coded space} such that\footnote{This condition ensures that each larger local code $\cF_\sigma$ satisfies all the constraints already imposed by the smaller ones $\cF_\tau$, $\tau\ge\sigma$.} for every $\sigma\le \tau$ from $X$ if $c\in\cF_\sigma$ then $c|_{X_\tau} \in \cF_\tau$.
\end{definition}

Now, for every open set $U\subseteq X$, we can define its code of \emph{local sections} as 
\[\cF(U) = \{c\in \F^{X_U} \mid  \forall \sigma\in U: c|_{X_\sigma}\in \cF_\sigma\}.\]

One can see that $\cF(X_{\ge\sigma})=\cF_\sigma$ for every $\sigma\in X$.
The reader can also easily check that the sets $\cF(U)$ we assign to open sets $U\subseteq X$ and the natural restriction maps $\cF_{U\to V}\colon c\mapsto c|_{X_V}$ satisfy conditions (1) and (2) from the general definition of sheaves in Subsection~\ref{ssc:topology}.

If in a~sheaf code all local codes are repetition codes over $\F$, then we denote it by $\underline{\F}(X)$ and call the~\emph{constant sheaf code on $X$ over $\F$} since for every open set $U$ the set $\underline{\F}(U)$ consists of all constant functions $c\colon X_U\to\F$.

We usually view $X$ as a~poset of cells of some cell complex, where the grading is the \emph{dimension} $\dim \sigma$ of a~\emph{cell} $\sigma\in X$, and $\dim X = \max_{\sigma\in X} \dim \sigma$ is called the \emph{dimension} of $X$. Furthermore, we say that $\cF_\sigma$ is a~\emph{level-$i$} local code, where $i = \dim X -\dim \sigma$. The level-$0$ local codes $\cF_\sigma$ assigned to the maximal elements ${\sigma\in X_*}$ are called \emph{trivial}, and we assume that $\cF_\sigma = \F^{\{\sigma\}}\cong \F$. 

We are mostly interested in codes $\cF$, which we call \emph{Tanner sheaf codes}, where  \emph{all} local codes are completely defined by the level-$1$ codes, i.e., 
\[
\cF_{\sigma}\! =\! \{ c\in \F^{X_\sigma}\! \mid c|_{X_\tau}\!\in\! \cF_\tau, \tau \ge \sigma, \dim \tau = \dim X - 1 \}.
\]

By default, we view codes $\cC \subseteq \F^S$ as sheaf codes, where the \emph{default coded space} for $\cC$ is the poset $\hat{S} = S\sqcup\{*\}$ with $\sigma < \tau$ \Iff $\sigma=*$ and ${\tau\in S}$, and the local codes $\cC_*= \cC$ and $\cC_{s_i}=\F^{\{s_i\}}\cong \F$ for $\sigma\in S = \{s_1,\dots,s_n\}$. Geometrically, the poset $\hat{S}$ can be represented as the $0$-dimensional simplicial complex $\{\varnothing, \{s_1\},\dots,\{s_n\}\}$ (\emph{empty graph}\footnote{Note that it is connected as a~finite topological space.} with $n$ vertices). 

Every classical LDPC code $\cC$ can be viewed as sheaf code on the poset $X$ of height~$1$, which Hasse diagram is the Tanner graph of the code and the level-$1$ local codes are single parity-check codes. Another example of sheaf codes are \emph{Tanner graph codes}~\cite{Tanner:1981}, such as Sipser-Spielman expander codes~\cite{Sipser:1996}, where the coded spaces are usually $\Delta$-regular graphs $X$, and the assigned to its vertices local codes are all equivalent to one \emph{base code}\footnote{If the base code $\cC_0$ is from some dedicated class of codes, then such codes are often called \emph{lifted codes} (e.g., see \cite{LiftingSmallLocally2019}). This way, one can consider \emph{lifted Reed-Solomon codes}~\cite{guoNewAffineinvariantCodes2013}, \emph{lifted projective Reed-Solomon codes}~\cite{lavauzelleLiftedProjectiveReed2019}, \emph{lifted tensor product codes}~\cite[Lemma~4.1]{dinurLocallyTestableCodes2022}.} $\cC_0\le \F^{\Delta}$.

The (\emph{tensor}) \emph{product code}\footnote{This definition is a~special case of the product of sheaf codes described in Section~\ref{sc:oper-sheaf}.} $\cC_1\otimes\dots\otimes\cC_D$ of codes ${\cC_i\le \F^{X_i}}$, $i\in [D]$ (see Appendix~\ref{sc:lp-codes}), can be also considered as a~sheaf code $\cC$ on $X = \hat{X}_1 \times \dots \times \hat{X}_D$, where $(x_1,\dots,x_D)\le (x'_1,\dots,x'_D)$ \Iff $x_i\le x'_i$, $i\in [D]$. Here the poset $X$ can be viewed as a~$(D-1)$-dimensional simplicial complex\footnote{In fact, it is the clique complex of the complete $D$-partite graph, and when $D=2$ it is just a~complete bipartite graph.}, where every $(x_1,\dots,x_m)\in X$ corresponds to the simplex 
\[\{(i, x_i) \mid x_i\ne * , i \in [D]\} \subseteq X_1\sqcup \dots \sqcup X_D.\]
We define the local codes for $\sigma=(x_1,\dots,x_D)\in X$ as 
\[\cC_\sigma = (\cC_1)_{x_1}\otimes\dots\otimes (\cC_D)_{x_D} \le \F^{X_1\times\dots\times X_D},\]
where every component code $\cC_i$ is considered as a~sheaf code $\cC_i=\cC_i(\hat{X}_i)$ on its default coded space $\hat{X}_i$.

\subsection{Maximally Extendable Sheaf Codes}\label{sc:main}

Given a~sheaf code $\cF$ on $X$, we say that a~local section $s\in\cF(U)$ for an~open set $U\subseteq X$ is \emph{extendable} if it can be extended to a~global section  $\hat{s}\in \cF(X)$ where $\cF_{X\to U}(\hat{s}) = s$; and we say that $U$ is \emph{extendable} if all its local sections are extendable.

\begin{definition}[ME sheaf codes]\label{df:MR-sheaf-code}
Given a class $\fF$ of sheaf codes\footnote{Note that our definition of maximal extendibility can be applied to sheaves with values in \emph{arbitrary} category, not necessary $\LC_\F$.} on a~space $X$, we say that $\cF\in \fF$ is \emph{maximally extendable} (\emph{ME}) for $\fF$  if every open set extendable in some code from $\fF$ is also extendable in $\cF$.  
\end{definition}

If $\vv=\{v_i\}_{i\in I}$ is a~set of free independent variables indexed by $I = I(\vv)$, then by $\F_p[\vv]$ and $\F_p(\vv)$ we denote respectively the ring of polynomials and the field of rational functions in variables $\vv$ over the prime field~$\F_p$. 
Consider a~sheaf code $\brF$ on $X$ over the field $\F_p(\vv)$, where all local codes $\brF_\sigma$ come with fixed \emph{polynomial} parity-check matrices $H_\sigma$ over the ring $\F_p[\vv]$. 
Now, given a~vector $\va\in\F_{q}^{I(\vv)}$, where $q=p^t$, we can substitute it instead of $\vv$ in all  $H_\sigma$ and get  the sheaf code $\brF[\va]$ over $\F_{q}$, called an~\emph{instance of $\brF$ over $\F_{q}$}, with local codes $\brF_\sigma(\va) = \ker H_\sigma(\va)$, $\sigma\in X$. Moreover, we say that $\brF[\va]$ is a~\emph{proper instance} if we also have $\dim \brF[\va] = \dim \brF$. Let $\brF[\F_q]$ be the set of all proper instances of $\brF$ over $\F_q$, and put $\brF[\cdot] = \cup_{\F_q\supseteq\F_p} \brF[\F_q]$.
\begin{definition}[generic code]
We say that $\brF$ is a~\emph{generic\footnote{This term is inspired by the notion of \emph{genericity} (e.g., see~\cite{brakensiekGenericReedSolomonCodes2023}). A very similar notion called \emph{topology} was earlier proposed in~\cite{gopalanExplicitMaximallyRecoverable2014}.} code for $\brF[\cdot]$} and call $\brF[\va]$ \emph{maximally extendable} for $\bar{\cF}$ if $\brF[\va]$ is maximally extendable for the class $\brF[\cdot]$.
\end{definition}

Let $S$ be a vector or a matrix over $\F(\vv)$. We say that $S$ is \emph{polynomial} if all the entries  are polynomials from $\F[\vv]$, and define $\deg S$ to be the maximal degree of polynomials from~$S$.
Let us first prove a~simple, but very useful lemma.
\begin{lemma}\label{lm:par-gen}
    If $H$ is a~polynomial parity-check matrix over $\F[\vv]$ for a~code $\cC$ of length $n$ over a~field $\F(\vv)$, then $\cC$ also has a~polynomial generator matrix\footnote{Here we do not assume that $H$ and $G$ are full rank.} $G$ such that  ${\deg G \le n\deg H}$\textup{;} and for all $\va\in \F^{I(\vv)}$, when $\rk H(\va) = \rk H$, we also have $\rk G(\va) = \rk G$.
\end{lemma}
\begin{proof}
 For every non-singular $m\times m$ minor $M$ in $H$, where $m=\rk H$, denote by $H_M$ the containing it $m\times n$ submatrix, and choose a~permutation $\pi\in \Sym_n$ such that $\pi(H_M) = [M, N]$, and $M$ is invertible. This gives us another parity-check matrix $\pi^{-1}([I_m, P])$ for $\cC$, where $P = M^{-1}N$. Next, let us define the generator matrix $G_M = \det M\cdot \pi^{-1}([-P^\T, I_{n-m}])$ for $\cC$, which is polynomial since $\det M\cdot M^{-1}$ is polynomial. From $\deg \det M\cdot M^{-1} \le (m-1)\deg M\le (m-1)\deg H$ we get $\deg G_M \le (m-1) \deg H + \deg H \le n \deg H$. Also, if $\det M(\va) \ne 0$ for $\va\in\F^{I(\vv)}$ then $\rk G_M(\va) = n-m$.
 
 Now, we combine the rows of $G_M$ for all such non-singular $m\times m$ minors $M$ in $H$ into one large generator matrix $G$ for $\cC$, where $\deg G\le n\deg H$. Finally, if for $\va\in \F^{I(\vv)}$ we have $\rk H(\va) = \rk H$, then $\det M(\va) \ne 0$ for some~$M$, and $\rk G(\va) = \rk G_M(\va) = \rk G$.
 \end{proof}

We are ready to state our main result\footnote{Its proof is inspired by the proof from~\cite[Lemma~32]{gopalanExplicitMaximallyRecoverable2014}.}. 

\begin{theorem}[ME sheaf codes exist]\label{th:main}
    For every generic code $\brF$ on $X$ over $\F_p(\vv)$ and ${p^t\! > q_{\brF}\! =\! d\abs{X}^2 2^{\abs{X}}}$, where $d = \max_{\sigma\in X}\deg H_\sigma$, there exists an~ME code $\brF[\va]$ over $\F_{p^t}$, and the probability that $\brF[\va]$ is not ME when we choose $\va\in \F_{p^t}^{I(\vv)}$ uniformly at random is bounded above as $q_{\brF}/p^t$. 
\end{theorem}
\begin{proof}
Consider a~generic code $\brF$ on $X$ over $\F_p(\vv)$, and for an~open set $U\subseteq X$ combine the rows from all $H_\sigma$, $\sigma\in U$, with the additional zero-padding if necessary, into one large parity-check matrix $H_U$ with $n_U = \abs{X_U}$ columns. It is not hard to see that $H_U$ is a~\emph{polynomial} parity-check matrix for the code $\brF(U)$. Applying Lemma~\ref{lm:par-gen} to $H=H_X$, we get a~\emph{polynomial} generator matrix $G$ over $\F_p[\vv]$ for the global code $\brF$, where $\deg G \le dn$,  $n=\abs{X_*}$. Moreover, for all \emph{proper} instances $\brF[\va]$, i.e., when $\rk H(\va) = \rk H$, the matrix $G(\va)$ is a~generator matrix for $\brF[\va] = \ker H(\va)$. 

Now for every open set $U\subseteq X$ we define the matrix $G_U = G|_{X_U}$, which is clearly a~generator matrix of the code consisting of all extendable sections from $\brF(U)$. Hence $U$ is extendable in $\brF$ \Iff $\im G^\T_U = \ker H_U$ \Iff 
\begin{equation}\label{eq:rk}
\rk G_U + \rk H_U = n_U.    
\end{equation}

Note that whenever $U$ is extendable for $\brF[\va]\in\brF[\cdot]$,  it~is also extendable in $\brF$ since $\im G^\T_U(\va) = \ker H_U(\va)$, and  from $\rk G_U(\va) + \rk H_U(\va) = n_U$ we get (\ref{eq:rk}) using\footnote{When we substitute $\va$ in a~matrix over $\F_p[\vv]$ we can not increase its rank since it is the maximal size of a~minor $M$ with $\det M \ne 0$.} $\rk G_U(\va) \le \rk G_U$, $\rk H_U(\va) \le \rk H_U$.

Obviously, condition (\ref{eq:rk}) is satisfied \Iff one can find an~${n_U\times n_U}$ non-singular minor $M_U$, i.e., $\det M_U \ne 0$, in the matrix $\rbr{\begin{smallmatrix}
H_U & 0\\
0 & G_U
\end{smallmatrix}}$. It is clear that 
\[\deg \det M_U \le n\deg G \le dn^2.\] 
Now let us fix for every extendable in $\brF$ open set $U\subseteq X$ such a~minor $M_U$ and denote by $f$ the product of $\det M_U$ for all such extendable $U$, which obviously has degree 
\begin{equation}\label{eq:q-bound}
\deg f \le dn^2 2^{\abs{X}} \le d\abs{X}^2 2^{\abs{X}}.     
\end{equation}

We claim that if $f(\va) \ne 0$ for $\va\in \F_{q}^{I(\vv)}$, ${q=p^t}$, then $\brF[\va]$ is maximally extendable for $\brF$. Indeed, let $f(\va) \ne 0$. We have $\brF[\va]\in\brF[\cdot]$ since any information set $U\subseteq X_*$ in $\brF$ is clearly extendable in it, and thus $\det M_U$ appears as a~factor in $f$. Hence $M_U(\va)\ne 0$ and $\rk G(\va) \ge \rk G_U(\va) = \abs{U} = \dim \brF$.
Moreover, if some open set $U$ is extendable in a~code from $\brF[\cdot]$, then $U$ is also extendable in $\brF$, and $\det M_U$ appears in~$f$. Thus, we get $\det M_U(\va)\ne 0$, and $U$ is extendable in $\brF[\va]$ since $\rk G_U(\va) + \rk H_U(\va) = n_U$, and therefore $\im G_U^\T(\va) = \ker H_U(\va)$.

Finally, by the Schwartz–Zippel lemma (e.g., see~\cite[Theorem~7.2]{motwaniRandomizedAlgorithms1995}), if $p^t > d\abs{X}^2 2^{\abs{X}} \ge \deg f$ then there exists $\va\in \F_{p^t}^{I(\vv)}$ such that $f(\va)\ne 0$. Therefore $\brF[\va]$ is maximally extendable, and the probability that this does not hold is bounded above as $\deg f/p^t$. 
\end{proof}

\begin{remark}
It is possible to define a~stronger form of maximum extendibility for a~code $\cF$ in a~family of sheaf codes $\fF$ on a space $X$ if for every two nested open sets $U\subseteq V$ we require that whenever in some code from $\fF$ every local section on $U$ is extendable to a~local section on $V$, then this also holds in $\cF$. With a~small modification in the proof of Theorem~\ref{th:main}, it is possible to prove its analogue for this stronger form of maximal extendibility if we put $q_{\brF} = d\abs{X}^2 2^{\abs{X}^2}$.    
\end{remark}

\section{Examples}\label{sc:examples}

In this section, for simplicity, we consider only the codes of even characteristic.
Let us illustrate the results from the previous section by giving some examples of generic sheaf codes on different complexes. We start from a~very simple example of a~generic code $\bar{\cC}$ over $\F_2(\vv)$ on the default coding space $X = \widehat{[n]}$ with only one non-trivial code $\bar{\cC}_*$ defined by the parity-check matrix $H = (v_{ij})_{m\times n}$ over $\F_2[\vv]$, where $\vv=(v_{ij})_{ij\in [m]\times[n]}$. Let $k=n-m$, it is clear that the set $\brC[\cdot]$ of all proper instances of $\bar{\cC}$ is exactly the set of all linear $[n, k]$ codes over all possible extension fields $\F_{2^t}$, which we denote by $[n,k]$. It is not hard to check that the extendable open sets $U\subseteq X$ in $\bar{\cC}$ are the subsets $U\subseteq [n]$, where $\abs{U}\le k$, and the set $X = \{*\}\cup [n]$. Indeed, if $*\in U$ then $U=X$ is clearly extendable. If $*\not\in U$ then $U\subseteq [n]$ is extendable \Iff $\dim \bar{\cC}|_U = \abs{U}$ since $\bar{\cC}(U) = (\F_p(\vv))^U$. Finally, for $U\subseteq [n]$ we have $\dim \bar{\cC}|_U = \abs{U}$ \Iff $\rk H|_{[n]\setminus U} = \rk H$ \Iff $\abs{U}\le k$. Hence, we see that $\cC\in [n,k]$ is maximally extendable for $\bar{\cC}$ \Iff for all $k$-element $U\subseteq [n]$ we have $\dim \cC|_U = k$. This means that the maximally extendable codes here are exactly the \emph{maximal distance separable} (\emph{MDS}) $[n, k]$ codes, which are usually defined by this property.

\subsection{Generic Tensor Product Codes}

Now, let us consider a~much more interesting example of a~generic code for the class $\bigotimes_{i\in [D]}[n_i,k_i]$ of product codes $\bigotimes_{i\in[D]}\cC_i$ over all possible fields $\F_{2^t}$, where we have $\cC_i\in [n_i,k_i]$. This example was first considered in~\cite{gopalanExplicitMaximallyRecoverable2014} in the context of maximally recoverable product codes.
We say that a~product code $\cC$ from the class $\bigotimes_{i\in[D]}[n_i,k_i]$ is \emph{maximally recoverable} if whenever $I$ is an~information set for some other product code $\cC'$ from this class, $I$ is also an~information set for $\cC$. This means that whenever we can uniquely \emph{recover} the erased symbols with indices $E\subseteq P = [n_1]\times\dots\times [n_D]$ in $\cC'$ since the unerased ones contain an~information set $I\subseteq P\setminus E$, then we can also recover them in $\cC$.

A~generic code for the class $\bigotimes_{i\in [D]}[n_i,k_i]$ can be defined as the tensor product code $\brC = \bigotimes_{i\in [D]}\brC_i$ over the field $\F_2(\vv)$, where each $\brC_i$ is a~generic code for $[n_i,k_i]$, defined, as in the example above, on the default space $\widehat{[n_i]}$ with a polynomial parity-check matrix $H^{(i)}$ over $\F_2[\vv^{(i)}]$, and $\vv= \sqcup_{i\in[D]} \vv^{(i)}$. Since $\deg H^{(i)} = 1$ for $i\in[D]$, then the parity-check matrix $H$ for the global code $\brC$ and the parity-check matrices $H_\sigma$ for all local codes $\brC_\sigma$, $\sigma \in X$, where $X = \widehat{[n_1]} \times \dots \times \widehat{[n_D]}$, all have degree $1$. Indeed, this is true since $H$ is obtained by stacking the matrices  $H^{(j)}=\otimes_{i\in[D]}M_i$, $j\in [D]$, of degree $1$, where $M_j=H_j$, and $M_i=\id_{n_i}$ for $i\ne j$. Hence, by Theorem~\ref{th:main} we have a~maximally extendable code for $\brC$  over $\F_q$ if $q > \abs{X}^2 2^{\abs{X}}$, where $\abs{X}=\prod_{i\in[D]} (n_i + 1)$.

\subsection{Generic Flag Product Codes}\label{ssc:flag-code}
In this section, we propose a~new class of sheaf codes called \emph{flag product codes}, generalizing the classical tensor product codes, and defined on certain flag complexes. Let us recall that a~\emph{flag complex} is a~simplicial complex\footnote{An (\emph{abstract})~\emph{simplicial complex} is a~down-closed family $X$ of sets $\sigma$ called its \emph{faces}. It is a~graded poset, where the grading $\dim \sigma = \abs{\sigma}-1$ is called the \emph{dimension}, and $\sigma$ is called a~$\dim \sigma$-face. Its $1$-skeleton $\sk_1 X$, called \emph{underlying graph}, consists of the set $V(X)$ of $0$-faces, called \emph{vertices}, and the set $E(X)$ of $1$-faces, called \emph{edges}.} $X$, where for every set of vertices $S$ pairwise connected by edges we have $S\in X$, i.e., it is completely defined by its underlying graph as the complex consisting of all its cliques. We only use flag complexes, which we call \emph{$(n_1,\dots,n_D)$-regular} or \emph{$n$-regular} when $n_i = n$, $i\in [D]$, obtained as simplicial complexes\footnote{The complex $\Delta(V)$ is usually called the \emph{order complex} of $V$.} $\Delta(V)$ consisting of all possible linearly ordered subsets called \emph{flags} from a~finite poset $V$ such that all \emph{maximal} by inclusion flags have size $D$, and every flag 
\[
\{x_1 < \dots < x_{i-1} < x_{i+1} < \dots < x_{D}\}
\]
is contained in exactly $n_i$ maximal flags. Since all maximal flags have the same size, $V$ is a~graded poset with $D$ levels $V_1,\dots,V_D$, where $V_i = V_i(X)$ is the set of $x\in V$ such that $i$ is the largest size of a~flag containing $x$ as its maximal element. It is very convenient for us to view $X$ as the subposet of $\hat{V}_1\times\dots\times\hat{V}_D$, identifying every flag $F\subseteq V$ with the~tuple $(x_1,\dots,x_D)$, where $x_i = v$ if $F\cap V_i = \{v\}$, and $x_i=*$ otherwise. 

For example, denote by $V$ the~set of vertices in the complete $D$-partite graph $K_{n_1,\dots,n_D}$, where the $i$-th part $V_i$ contains $n_i$ vertices. If we define the order in $V$  as $x < y$ \Iff $x\in V_i$, $y\in V_j$, and ${i < j}$, then the complex $\Delta(V)$ is just the \emph{clique complex} of $K_{n_1,\dots,n_D}$, i.e., its faces are the cliques in this graph. Hence, we see that such complexes $\Delta(V)$ are exactly the coded spaces ${\hat{V}_1\times\dots\times\hat{V}_D}$ of product codes (see Section~\ref{sc:tanner-sheaf}).

Another important example is when the poset $V$ is the set of all proper subspaces  $\cV\le\F_q^{D+1}$, $1\le \dim \cV \le D$, naturally ordered by inclusion, where we denote $\Delta(V)$ by $A_{D-1}(\F_q)$ and call the \emph{flag complex}\footnote{Such complexes are often called \emph{flag varieties}.} of $\F_q^{D+1}$. We see that $\dim A_{D-1}(\F_q) = D-1$, and, in particular, one can easily recognize that $A_1(\F_q)$ is the well-known graph of points versus lines in the projective plane $\PP^2_q$ over $\F_q$.  

Already in~\cite{Tanner:1981}, Tanner proposed a~code on the graph $A_1(\F_q)$ called the \emph{field plane hexagon} (\emph{FPH}) code. Another example of a~code on $A_1(\F_q)$ is the code defined by Rodier in~\cite{rodierCodesFlagVarieties2003}, as a~polynomial evaluation code on the set of pairs $(\alpha,\beta)\in \PP^2_q\times \PP^2_q$ such that $\dtp{\alpha}{\beta} = 0$. The last code can also be considered as the~punctured product code initially defined on \emph{all} pairs $(\alpha,\beta)\in \PP^2_q\times \PP^2_q$, where we then puncture the symbols with $\dtp{\alpha}{\beta}\ne 0$. Here we adopt a similar approach. 

Consider some $(n_1,\dots,n_D)$-regular flag complex $X$, and chose $D$ codes $\cC_i \le \F^{V_i}$ over $\F=\F_{2^t}$, $i\in [D]$, where $V_i = V_i(X)$. Let us define the sheaf code $\cF=\cF(X;\cC_1,\dots,\cC_D)$ on $X$, called a~\emph{flag product} (\emph{FP}) \emph{code}, with local codes $\cF_\sigma = \cC|_{X_\sigma}=\{c|_{X_\sigma}\mid c\in \cC\}$ obtained from $\cC = \bigotimes_{i\in [D]} \cC_i\le \F^{V_1\times\dots\times V_D}$, $\sigma\in X$, where  $\cC_1,\dots,\cC_D$ are called the \emph{component codes} of~$\cF$. 
Each component code $\cC_i$ is an~$[N_i,k_i]$ code over $\F$, where $N_i=\abs{V_i}$.
The~poset $X$ has the greatest element $\sigma=(*,\dots,*)$, hence $\cF_\sigma$ coincides with the global code $\cF(X) =\cC|_{X_*}$, where $X_*$ is the set of all maximal flags. We can see that when $X$ is the clique complex of the complete $D$-partite graph, then $\cF$ is the standard product code $\cC\!=\!\bigotimes_{i\in[D]}\cC_i$, while in the general case it is obtained from $\cC\le \F^{V_1\times\dots\times V_D}$ by puncturing the indices $\sigma\not\in X_*$.

If every $\cC_i$ is an~MDS code, and $k_i=\dim \cC_i \le n_i$, then we have a~very simple formula for the dimension
\begin{equation}\label{eq:fp-dim}
    \dim \cF = k_1\cdots k_D.
\end{equation}
Indeed, since $\dim \cC = k_1\cdots k_D$, we just need to show that $\dim \cC = \dim \cC|_{X_*}$, i.e., the dimension is preserved after the puncturing. However, since all component codes are MDS and $k_i\le n_i$, we can uniquely recover $N_i - n_i$ punctured symbols for each level-$1$ local code $\cC_\sigma$.

Now if $\brC = \bigotimes_{i\in [D]}\brC_i$ is the generic product code with the standard polynomial parity-check matrix $H$ over $\F_2[\vv]$ described in the previous subsection,  where $\brC_i$ is the generic code on the default coded space $\hat{V}_i$, then we can also define a~generic FP code $\brF$ on the flag complex $X$ if for each $\sigma\in X$ we can find a~polynomial parity-check matrix $H_\sigma$ such that $\ker H_\sigma = \brC|_{X_\sigma}$. We can do this if we first find a~polynomial generator matrix $G$ for $\brC=\ker H$ using Lemma~\ref{lm:par-gen}. Then we see that $G_\sigma = G|_{X_\sigma}$ is a~polynomial generator matrix for $\brC|_{X_\sigma}$, and, finally, we can also find a~polynomial parity-check matrix $H_\sigma$ for it using Lemma~\ref{lm:par-gen} once again\footnote{In this case, we apply Lemma~\ref{lm:par-gen} to the dual code $(\brC|_{X_\sigma})^\perp$.}. It is clear that since we used Lemma~\ref{lm:par-gen} twice, we have ${\deg H_\sigma \le \abs{X}^2\deg H \le \abs{X}^2}$, where $\abs{X}=(N_1+1)\cdots (N_D+1)$, and from Theorem~\ref{th:main} we have that there exists a~maximally extendable FP code over $\F_{q}$ if $q > |X|^4 2^{\abs{X}}$. Moreover, since in the generic product code $\brC$ all component codes $\brC$ are MDS, then its dimension is given by (\ref{eq:fp-dim}), which is also true for all its proper instances, including maximally extendable FP codes.

Let us consider, as an~example, maximally extendable FP codes $\cF=\cF(X;\cC_1,\cC_2)$ on $X=A_1(\F_p)$, where $p$ is an~odd prime number, and $\dim \cC_i = (p+1)/2$, ${i\in [2]}$. The component codes $\cC_1$ and $\cC_2$ are $[p^2+p+1, (p+1)/2]$ codes, and therefore $\dim \cF = (p+1)^2/4$.  This example may be relevant in the context of Question~5.7 from~\cite{dinurGoodLocallyTestable2022} asking whether one can define LTCs on $A_1(\F_p)$. 

\begin{remark}
    Though we do not discuss it here, it is also possible to define generic codes for FHP codes proposed by Tanner~\cite{Tanner:1981}, and show, using Theorem~\ref{th:main}, the existence of maximally extendable FHP codes.
\end{remark}

\section{Operations with Sheaf Codes}\label{sc:oper-sheaf}

One can naturally apply the standard operations with cellular sheaves~\cite{hansenSpectralTheoryCellular2019}, \cite{firstGoodQueryLocally2023} such as product, pullback, etc. to sheaf codes. Below we show how they can be adapted.

Let us start from a~very simple and basic operation. For a given sheaf code $\cF$ on $X$, we can define its \emph{restriction} $\cF|_U$ to a~subposet $U\subseteq X$, as the sheaf code on the space $U$ with the global code $\cF(U)$ and the local codes $\cF_\sigma$, $\sigma\in U$. Since every set $X_{\ge\sigma}$ is open, we can identify the local code $\cF_\sigma = \cF(X_{\ge\sigma})$ with the sheaf code $\cF|_{X_{\ge\sigma}}$.

Now, let us show how to glue together several small sheaf codes into the larger one. Suppose we have a~poset $X$ with a~fixed \emph{open cover}, i.e., with a~family of open sets $\{U_i\}_{i\in I}$ such that $X=\cup_{i\in I} U_i$. Now, if for every $i\in I$ we have a~sheaf code $\cF_i$ on $U_i$, we can define the \emph{union} $\cup_{i\in I} \cF_i$ as the sheaf code $\cF$ on $X$, where for every ${\sigma\in X}$ we put\footnote{Recall that  $\sum_{i\in I}\cC_i = \{\sum_{i} c_i \mid c_i\in \cC_i, i\in I\}$.} $\cF_\sigma = \sum_{i\in I} (\cF_i)_\sigma$. The definition is correct since for $\sigma\le\tau$ if $c= \sum_{i\in I} c_i\in \cF_\sigma$, where $c_i\in (\cF_i)_\sigma$, then $c|_{X_\tau} = \sum_{i\in I} c_i|_{X_\tau}\in \cF_\tau$. Note that when all codes are \emph{compatible}, i.e., $\cF_i|_{U_i\cap U_j} = \cF_j|_{U_i\cap U_j}$, $i,j\in I$, then it is clear that $\cF$ is compatible with all of them, i.e., $\cF|_{U_i} = \cF_i$ for $i\in I$.  Moreover, if $X=\sqcup_{i\in I} U_i$, then $\cF(X) = \bigoplus_{i\in I} \cF_i(U_i)$.
It is also clear we have the following easy to verify lower bound for the minimal distance $d(\cup_{i\in I}\cC_i)\ge \min_{i\in I} d(\cC_i)$.

Next, let us move to more interesting operations.
\begin{definition}[product of sheaf codes]
The \emph{product} $\cA \otimes \cB$ of sheaf codes $\cA(X), \cB(Y)$ is the sheaf code $\cC$ defined on $X\times Y$ as $\cC_{\sigma\times \tau} = \cA_\sigma \otimes \cB_\tau \le \F^{X_\sigma\times Y_\tau}$; $\sigma\in X,\tau\in Y$.  
\end{definition}
The definition is correct since if $c\in \cA_\sigma \otimes \cB_\tau$ then $c|_{X_{\sigma'} \times Y_{\tau'}}\in \cA_{\sigma'} \otimes \cB_{\tau'}$ for every $(\sigma, \tau) \le (\sigma', \tau')$, which also shows us the restriction maps in this sheaf 
\begin{equation}
\cC_{(\sigma, \tau) \to (\sigma', \tau')} = w \mapsto w|_{X_{\sigma'} \times Y_{\tau'}}.  
\end{equation}

It is easy to check that the global code $\cC(X\times Y)$ is exactly the tensor product code $\cA(X) \otimes \cB(Y)$ of the global codes $\cA(X), \cB(Y)$, which justifies the notation and the name we used here. 

Now let us consider two general operations that allow one to obtain large codes out of smaller ones and vice versa. For an~arbitrary map $\phi\colon X\to Y$, we call the~code 
\[
\phi^* \cB = \{b\circ\phi\in \F^X \mid b\in \cB\}
\]
the~\emph{pullback} of $\cB\le \F^Y$. Similarly, we call the code
\[
\phi_* \cA = \{b\in \F^Y \mid b\circ\phi \in \cA\}
\]
the \emph{pushforward}\footnote{The pullback and the pushforward are also called the \emph{inverse} and the \emph{direct image}, respectively.} of $\cA\le \F^{X}$.

Most often, $\phi$ is the natural projection $\phi\colon X\to X/{\sim}$ for some equivalence relation $\sim$ on a~set~$X$, in which case we usually call $\cB$~the \emph{quotient} of $\cA$. The most important for us case is when $X/{\sim} = X/G$ is the orbit space for an~action of a~group~$G$ on $X$ (see Appendix~\ref{sc:action}). 

Let us show how we can adapt all this for sheaf codes.
Given graded posets $X$ and $Y$, we say that $\phi\colon X\to Y$ is a~\emph{morphism} if for every $\sigma\le \tau$ we have $\phi(\sigma) \le \phi(\tau)$ and $\dim \phi(\sigma) = \dim \sigma$. 
\begin{definition}[pullback of sheaf codes]
If $\phi\colon X\to Y$ is a~morphism and $\cB$ is a~sheaf code on $Y$, the \emph{pullback} is the~sheaf code $\phi^*\cB$ on $X$, where for every $\sigma\in X$ we put\footnote{Note that $\phi(X_{\sigma}) = Y_{\phi(\sigma)}$ since $\phi$ is a~morphism.} $(\phi^*\cB)_\sigma = \{b\circ\phi|_{X_\sigma} \in \F^{X_\sigma} \mid c \in \cB_{\phi(\sigma)}\}$.
\end{definition}

Here the definition is correct since for every $\sigma\le \tau$ if $b\circ\phi|_{X_\sigma}\in (\phi^*\cB)_\sigma$ where $ b \in \cB_{\phi(\sigma)}$ then we have $(b\circ\phi|_{X_\sigma})|_{X_\tau} = b\circ\phi|_{X_\tau} = b|_{Y_{\phi(\tau)}}\circ\phi|_{X_\tau}\in (\phi^*\cF)_\tau$.  

Pullbacks are convenient when one starts from a~base code $\cB$ defined on a~small complex $X$ and then \emph{lifts} it to a~large code $p^*\cB$ defined on its $\ell$-fold covering $\hat{X}$ using the corresponding covering map $p\colon \hat{X} \to X$. This way, one can describe the \emph{classical} lifted product codes from~\cite{Panteleev&Kalachev:2021:ltc, dinurGoodQuantumLDPC2023} (see also \cite{firstGoodQueryLocally2023} for different examples). 

Another example is a~Tanner graph code~\cite{Tanner:1981} defined on a~$\Delta$-regular graph $\Gamma$ with some fixed labelling map $\lambda\colon E(\Gamma)\to [\Delta]$, where we assign to every vertex a~small base code $\cB\in\F^\Delta$. Here, we can define the code as the pullback $\phi^*\cB$, where we consider $\cB$ as sheaf code on the default coded space $\widehat{[\Delta]}$, and the morphism $\phi\colon \Gamma\to \widehat{[\Delta]}$ is defined as $\phi(v) = *$ and $\phi(e) = \lambda(e)$ for ${v\in V(\Gamma)}$, ${e\in E(\Gamma)}$. This more general idea is also often called a~\emph{lifting} of the base~code $\cB$~(e.g., see \cite{LiftingSmallLocally2019}).     

\begin{definition}[pushforward of sheaf codes]
If $\phi\colon\! X\to Y$ is a~morphism and $\cA$ is a~sheaf code on $X$, the \emph{pushforward} is the~sheaf code $\phi_*\cA$ on $Y$, where for every $\sigma\in Y$ we put
\[(\phi_*\cA)_\sigma = \{b \in \F^{Y_\sigma} \mid b\circ\phi|_{\phi^{-1}(Y_\sigma)} \in \cA({\phi^{-1}(Y_{\ge\sigma})})\}.\]
\end{definition}
The definition is correct since for $\sigma\le\tau$ we have $Y_{\ge\tau} \subseteq Y_{\ge\sigma}$, and thus $\cA({\phi^{-1}(Y_{\ge\sigma})}) \le \cA({\phi^{-1}(Y_{\ge\tau})})$.
An important example of the pushforward is the \emph{quotient} of sheaf codes, invariant under a~group $G$ acting on the coded space $X$. Here we have a~natural morphism $\phi\colon X\to X/G$, and thus can define the quotient $\cA/G$ as $\phi_*\cA$. We give this definition in Appendix~\ref{sc:action}, where we also define balanced products of sheaf codes, which can be viewed as the~quotient of the product of sheaf codes. 

Yet another very interesting operation can be applied when $\cF$ is a~Tanner sheaf code,
in which case, we can define its \emph{dual sheaf code} $\cF^\T$ as code on the same coded space $X$, where all level-$1$ codes are replaced by their duals, i.e., we have $(\cF^\T)_\sigma = \cF_\sigma^\perp$ for all $\sigma \in X$ with $\dim \sigma = \dim X - 1$. Such codes naturally arise when one considers the codes $\ker H^\T$, where $H$ is a parity-check matrix of a~Tanner code~\cite{Panteleev&Kalachev:2021}, \cite{Breuckmann:balanced:2021} (see also~\cite{firstGoodQueryLocally2023, dinurNewCodesHigh2023, dinurqLTC:2024}).
It is not hard to verify that  $(\cF^\T)^\T = \cF$ and $(\cA\otimes\cB)^\T = \cA^\T \otimes \cB^\T$. Moreover, if $\cF$ is a~\emph{default sheaf code}, i.e., its coded space is a~default space $X=\hat{S}$, then for the global code we have $\cF^\T(X) = \cF^\perp(X)$. For the product code, considered as a~Tanner sheaf code, we clearly have $(\bigotimes_{i\in[D]} \cC_i)^\T = \bigotimes_{i\in[D]} \cC_i^\perp$.

In sheaf theory there is a~dual notion to sheaf called \emph{cosheaf} where we also assign vector spaces $\cF(U)$ to open sets $U\subseteq X$, but for every pair of nested open sets $U \subseteq V$ instead of a~restriction map we define an~\emph{extension map} $\cF_{U \to V}\colon \cF(U)\to \cF(V)$. Thus we can also consider cosheaves of linear codes. If $\cF$ is a~sheaf code on $X$, we can define its \emph{dual cosheaf code} $\cF_\sigma^\perp$ obtained from $\cF$ if we replace all local codes $\cF_\sigma$ by their duals $\cF_\sigma^\perp$. Here we have an~extension map $c\mapsto c|^{X_\sigma}$, where for every $\sigma\le \tau$ we extend a~codeword $c\in\F^{X_\tau}$ by zeros to the codeword $c'=c|^{X_\sigma}\in\F^{X_\sigma}$ such that $c'(x)=c(x)$ if $x\in X_\tau$, and $c'(x)=0$ otherwise. An~example of a~cosheaf code is the dual product code $\cA_1\boxplus \dots \boxplus \cA_D$ defined as $(\cA_1^\perp \otimes \dots\otimes \cA_D^\perp)^\perp$.

\section{Cohomology of Sheaf Codes}\label{sc:cohom}

Since every sheaf code $\cF$ comes together with its coded space $X$, it is possible to study their cohomology groups, denoted as $H^i(\cF)$ or $H^i(X; \cF)$, and call the \emph{cohomology groups  of $X$ with coefficients in $\cF$} (e.g., see~\cite{hansenSpectralTheoryCellular2019}). In the special case of constant sheaf codes $\underline{\F}(X)$, this is exactly the standard cohomology $H^i(X; \F)$ of $X$ with coefficients in $\F$. 

Though for arbitrary posets $X$ one can use the \emph{\v{C}ech cohomology}~\cite{currySheavesCosheavesApplications2014} to define the cohomology group of $X$, here we restrict ourselves to the simpler case~\cite{hansenSpectralTheoryCellular2019} when $X$ is the \emph{cell poset}\footnote{Such posets are often called \emph{CW posets}~\cite{bjorner1984posets}.}, i.e., the poset of cells of a~regular cell complex (e.g., simplicial or cubical complex), where the cells are naturally ordered by inclusion. All cell posets have a~very nice property~\cite[Fig.~1]{bjorner1984posets} that the number of $i$-cells between every $(i-1)$-cell and $(i+1)$-cell is either $0$ or $2$.  
Let $\cF$ be a~sheaf code on a~cell poset $X$ over a~field $\F$ of characteristic\footnote{For simplicity here we consider only the case $\ch \F = 2$. In the general case, the definition of coboundary maps is essentially the same, but we need to take into account the orientation of $X$.}~$2$. 
Consider the alternating sequence of vector spaces over $\F$ and $\F$-linear maps called a~\emph{cochain complex} 
\[
\cdots\rightarrow\bC^{i-1} \xrightarrow{\delta^{i-1}}  \bC^i \xrightarrow{\delta^{i}}  \bC^{i+1} \rightarrow \cdots 
\]
where\footnote{One can identify $\bigoplus_{i\in I}\cC_i$ with $\{(c_i)_{i\in I} \mid  c_i\in\cC_i, i\in I \}$, and if $I=\varnothing$ we assume that we get the zero vector space $\bz$.} $\bC^i = \bigoplus_{\sigma\in X(i)} \cF_\sigma$, $i\in\Z$, and the $\F$-linear maps, called \emph{coboundary maps}, are defined as
\[
\delta^i (c_{\sigma}\cdot\sigma) = \sum_{\sigma \prec \tau} (c_{\sigma}|_{X_\tau})\cdot\tau, 
\]
and extended by linearity.
We represent the elements of the vector space $\bC^i = \bC^i(X; \cF)$ as formal linear combinations $\sum_{\sigma\in X(i)} c_\sigma\cdot \sigma$, where $c_\sigma\in \cF_\sigma$. Since the natural restriction maps $\cF_{\sigma\to\tau}\colon x\mapsto x|_{X_\tau}$ satisfy the condition $\cF_{\tau\to\pi} \circ \cF_{\sigma\to\tau} = \cF_{\sigma\to\pi}$  (see also Definition~\ref{df:cellular-sheaf}) and the cell poset $X$ has the nice property mentioned above, it is straightforward to check (e.g., see~\cite{hansenSpectralTheoryCellular2019}) that $\delta^{i}\circ\delta^{i-1} = 0$. This implies that $\im \delta^{i-1} \le \ker\delta^{i}$, and allows one to define the \emph{$i$-th cohomology group} $H^i(\cF) = H^i(X; \cF)$ as $\bZ^i/\bB^i$, where $\bB^i = \im \delta^{i-1}$ and $\bZ^i = \ker\delta^{i}$ are called respectively the spaces of \emph{$i$-cocycles} and \emph{$i$-coboundaries}. We also want to note that the cochain complex can be viewed as \emph{one} large vector space $\bC^\bullet = \bigoplus_{i} \bC^i$ with \emph{one} coboundary map $\delta\colon \bC^\bullet\to\bC^\bullet$ such that $\delta^i = \delta|_{\bC^i}$.

The cohomology groups $H^i(\cF)$ allow one to define \emph{quantum} codes out of classical sheaf codes $\cF$. For example, if we fix a~generator matrix $G$ for the global code $\cF(X)$ we can equip with bases $\tilde{\cF}_\sigma$ all local codes $\cF_\sigma$ if we define $\tilde{\cF}_\sigma$ to be the set of first linearly independent rows in the matrix $G|_{X_\sigma}$. Therefore, every element of $\bC^i$ is uniquely represented as a~vector of length $\sum_{\sigma\in X(i)} \dim \cF_\sigma$, and the quotient $H^i(\cF) = \bZ^i/\bB^i$ can be viewed as a quantum CSS code of this length.

\section{Expansion of Sheaf Codes}\label{sc:exp-sheaf-codes}

If we equip the cochain complex $\bC^\bullet = \bC^\bullet(X;\cF)$ with a~norm, we can naturally transfer all the standard notions, such as coboundary expansion and local minimality, from theory of high-dimensional expanders (HDXs)~\cite{Kaufman:2014, Lubotzky:2018}, originally developed for cohomology groups over $\F_2$ to the much more general realm of cohomology groups with local coefficients in the~sheaf~$\cF$. 

Here we use the \emph{block} (\emph{Hamming}) \emph{norm}\footnote{This norm is also sometimes called the \emph{support norm}~\cite{firstGoodQueryLocally2023}.} from~\cite{Panteleev&Kalachev:2021:ltc}, naturally generalizing the standard Hamming norm, defined for $c\in \bC^\bullet$ as $\abs{c}_X = \abs{\{\sigma\in X\mid c_\sigma \ne 0\}}$, and also consider its normalized version\footnote{This normalization is often used for simplicial complexes~\cite{Lubotzky:2018}.} as $\norm{c}_X = \sum_{c_\sigma\ne 0} w_{\sigma}$, where $w_\sigma = \abs{X_\sigma}/\sum_{\sigma\in X(i)}\abs{X_\sigma}$, $\dim \sigma = i$. In fact, in all examples we are interested, the poset $X$ is sufficiently regular, and  $w_\sigma = 1/\abs{X(i)}$ for $\sigma\in X(i)$. Both these norms can be naturally extended to subspaces $S\le \bC^\bullet$ as $\abs{S}_X = \min_{\sigma\in S} \abs{\sigma}_X$ and $\norm{S}_X = \min_{\sigma\in S} \norm{\sigma}_X$. Now, we are ready to define the \emph{$i$-th coboundary expansion} $\h^i(\cF) = \h^i(X;\cF)$ of $\cF$ as 
\begin{equation}
    \h^i(X;\cF) = \min_{c\in \bC^i\setminus \bB^i} \frac{\abs{\delta^{i} c}_X}{\abs{c + \bB^i}_X},
\end{equation}
and also its \emph{normalized} version $\tilde{\h}^i(\cF) = \tilde{\h}^i(X; \cF)$, where instead of $\abs{\cdot}_X$ we use $\norm{\cdot}_X$. Here we assume that the minimum over the empty set is equal to $\infty$, and also put $\h(\cF) = \min_{i\in\Z} \h^{i}(\cF)$ and $\tilde{\h}(\cF) = \min_{i\in\Z}\tilde{\h}^i(\cF)$. Note that it is also possible to define the \emph{small set} versions~\cite{hopkinsExplicitLowerBounds2022a} of all these quantities: 
\begin{equation}
    \h^i_{a}(X;\cF) = \min_{\substack{c\in \bC^i\setminus \bB^i\\ \abs{c}\le a}} \frac{\abs{\delta^{i} c}_X}{\abs{c + \bB^i}_X}.
\end{equation}

In the special case when $\cF$ is the constant sheaf $\underline{\F_2}$ on a~simplicial complex $X$, we get exactly the coboundary expansion from the theory of HDXs, which in turn generalizes the Cheeger constant $h(X)$ (e.g., see~\cite{Lubotzky:2018, firstGoodQueryLocally2023}) of a~graph $X$ since we have $h(X) = \h^0(X; \underline{\F_2})$ when we view graphs as simplicial complexes\footnote{This assumption is very important here since simplicial complexes also have the empty face $\varnothing$ with $\dim \varnothing = -1$ (e.g., see \cite{Lubotzky:2018}).}.

At the same time, when $\cF$ is a~default sheaf code of length $n$, then $\h(\cF)$ and $\tilde{\h}(\cF)$ coincide respectively with the minimum distance $d(\cF)$ and its normalized version $\frac{1}{n} d(\cF)$. In the case of tensor product code $\cA \otimes\cB$, we have that $\h(\cA \otimes\cB)$ is, up to normalization factors, the \emph{agreement testability}~\cite{dinurHighDimensionalExpanders2017}, which is in turn related with the \emph{robust testability}~\cite{Ben-Sasson:2006} of $\cA \otimes\cB$ (e.g., see~\cite{dinurGoodQuantumLDPC2023, kalachevTwosidedRobustlyTestable2023}). 

The quantity $\h(\cA \otimes\cB)$ was an important ingredient in the first constructions of good LTCs and qLDPC codes~\cite{Dinur:2021, Panteleev&Kalachev:2021:ltc}, where it played a~role similar to the minimal distance in the Sipser-Spielman codes.  
Note, that in~\cite{Panteleev&Kalachev:2021:ltc} it was called \emph{product expansion}, and also considered in the case when we have some puncturing in the product code $\cA\otimes\cB$. Later Leverrier and Z{\'e}mor found an~elegant equivalent reformulation of the product expansion, and finally in \cite{dinurGoodQuantumLDPC2023, kalachevTwosidedRobustlyTestable2023} it was directly connected to the (co)boundary expansion, and also generalized to $D>2$ dimensions. It is important to note that when $D>2$ it is a~strictly \emph{stronger} property than both the agreement and the robust testability~\cite{Kalachev:example:2023}. Finally, in~\cite{firstGoodQueryLocally2023} the authors studied coboundary expansion for arbitrary sheaves of vector spaces on simplicial complexes, and also used it to construct LTCs and qLDPC codes.

Let us also say that a~sheaf code $\cF$ is an~\emph{$\eta$-expander} if $\eta(\cF)\ge \eta$, and a~\emph{two-way $\eta$-expander} if both $\cF$ and $\cF^\T$ are $\eta$-expanders, where $\cF^\T$ is the dual sheaf code. Note that if $\cF=\bigotimes_{i\in[D]}\cC_i$, then two-way $\eta$-expansion requires that both $\bigotimes_{i\in[D]}\cC_i$ and $\bigotimes_{i\in[D]}\cC_i^\perp$ are $\eta$-expanders, and when $D=1$ it requires that the minimal distance of a~code and its dual is at least $\eta$. The case $D=1$ was essential for a~construction of Sipser-Spielman codes with good bidirectional small-set expansion properties in~\cite{Panteleev&Kalachev:2021}, \cite{Breuckmann:balanced:2021}. At the same time, the case $D=2$ was essential for expander LP codes~\cite{Panteleev&Kalachev:2021:ltc} and their simplification called \emph{quantum Tanner codes}~\cite{leverrierQuantumTannerCodes2022a} to get the linear minimum distance, and also good bidirectional small-set expansion~\cite{hopkinsExplicitLowerBounds2022a}, which in turn implies the NLTS theorem~\cite{anshuNLTSHamiltoniansGood2023}.

\section{Conclusions}

We proposed a~general framework of sheaf codes, inspired by the Dinur-Kauffman sheaf theoretic approach to property testing~\cite{dinurHighDimensionalExpanders2017} and the Meshulam's idea~\cite{Meshulam:2018} to define Tanner codes on simplicial complexes as local systems of coefficients. This general framework complements the existing literature on the codes defined on various types of high-dimensional expanders~\cite{Dinur:2021, Panteleev&Kalachev:2021:ltc, kalachevTwosidedRobustlyTestable2023, firstGoodQueryLocally2023, dinurNewCodesHigh2023, dinurqLTC:2024}, and provides a~topological interpretation of such codes as sheaves on finite topological spaces, or more generally on Alexandrov topological spaces. Though sheaf codes can also be viewed as cellular sheaves of vector spaces~\cite{currySheavesCosheavesApplications2014, hansenSpectralTheoryCellular2019, firstGoodQueryLocally2023} or, equivalently, local systems on posets~\cite{Panteleev&Kalachev:2021:ltc}, considering them as sheaves on topological spaces allows us to study new properties such as maximal extendibility, which is in turn closely related with maximally recoverable codes~\cite{chenMaximallyRecoverableProperty2007, gopalanExplicitMaximallyRecoverable2014}. Note also that the problem of extension to global sections, which we study here, also naturally appears in the context of network coding~\cite{NetSheaf:2011}. 

We showed that in every class of sheaf codes defined on the same space $X$ and admitting a~polynomial parametrization of the parity-check matrices for the local codes, there always exists a~maximally extendable code over a~sufficiently large finite field. In particular, this implies the existence of maximally extendable product codes with \emph{arbitrary} number of component codes. This is important since such codes have good coboundary expansion properties~\cite{kalachev:2024}, and, as it was suggested in~\cite{kalachevTwosidedRobustlyTestable2023}, 
tensor products of more than two codes with good coboundary expansion could potentially be used to construct multidimensional versions of the codes from~\cite{Panteleev&Kalachev:2021:ltc, leverrierQuantumTannerCodes2022a, dinurGoodQuantumLDPC2023}, and thus get good qLTCs. 

In the~recent result~\cite{dinurqLTC:2024}, it is shown that using a~natural $4$-dimensional generalization of the codes from~\cite{Dinur:2021, Panteleev&Kalachev:2021:ltc} one can get \emph{almost} good qLTCs provided that the product of the four local codes has good two-way coboundary expansion. 
However, we believe that to get good qLTCs it is required to consider significantly different complexes than the ones used in~\cite{Dinur:2021, Panteleev&Kalachev:2021:ltc, dinurqLTC:2024}.   
We think that an interesting future research direction is to use the Ramanujan cubical complexes from~\cite{rungtanapiromInfiniteSeriesQuaternionic2019} obtained as the~quotient $X=T_1\times\dots\times T_D/G$ of a~product of $D$ infinite regular trees $T_1,\dots, T_D$, where $G$ is a~group acting on this product. We conjecture that with an~appropriate choice of local product codes\footnote{Tensor product of the \emph{projective} Reed-Solomon codes could be a~good choice to get good LTCs using these complexes since the automorphism groups of such codes contain the local permutation groups of the complexes from~\cite{rungtanapiromInfiniteSeriesQuaternionic2019}. Unfortunately, a~dual of a~projective Reed-Solomon code is still a~projective Reed-Solomon code, and such codes can not be used~\cite{kalachevTwosidedRobustlyTestable2023} in the current constructions of good qLDPC codes. However, we believe that if one can find maximally expandable product codes, where component codes are also stable under such permutations than they can be used to construct good qLTCs.}, it is possible to define a~sheaf code $\cF$ on this complex $X$ and the obtained qLDPC code $H^2(X; \cF)$ could potentially give a~positive solution to the qLTC conjecture.

In Section~\ref{ssc:flag-code} we consider sheaf codes defined on flag complexes.
A~potential use of such codes is a~natural generalization of quantum color codes~\cite{color2006}, which we call \emph{Tanner color codes}, inspired by the idea of quantum Tanner codes from~\cite{leverrierQuantumTannerCodes2022a}. If we have a~Tanner sheaf code $\cF$ on a~$2$-dimensional $3$-partite simplicial complex $X$, we can consider its dual sheaf code $\cF^\T$ (see Section~\ref{sc:oper-sheaf}), where $\cF^\T_e = \cF_e^\perp$, for all $e\in X(1)$, and define the~quantum CSS code $Q(\cF)=C_X/C_Z^\perp$, with the stabilizer groups $C_X^\perp = \sum_{v\in X(0)} \cF_v$ and $C_Z^\perp = \sum_{v\in X(0)} \cF^\T_v$ generated  respectively by the codewords from $\cF_v$ and $\cF^\T_v$, ${v\in X(0)}$. Since for \emph{different} vertices $v,v'\in X(0)$ we have that either $X_v \cap X_{v'} = \varnothing$ or  $X_v \cap X_{v'} = X_{e}$, where $e=\{v,v'\}\in X(1)$, then $\dtp{c}{c'} = \dtp{c|_{X_e}}{c'|_{X_e}} = 0$ for all $c\in\cF_v$ and $c'\in\cF^\T_{v'}$ because $\cF^\T_e = \cF_e^\perp$. Now, when $v=v'$, since $X$ is $3$-partite, we have that the link $X_{\ge v}$ is a~$2$-partite graph with parts $E_0,E_1\subseteq X(1)$, and thus $X_{v} = \sqcup_{e\in E_0} X_e$. Hence, we still have $\dtp{c}{c'} = \sum_{e\in E_0} \dtp{c|_{X_e}}{c'|_{X_e}} = 0$ in this case, and the definition of $Q(\cF)$ is indeed correct. 

One can consider defining Tanner color codes $Q(\cF)$ on simplicial complexes from~\cite{dinurNewCodesHigh2023}, \cite{Kaufman:2014}, \cite{kaufmanConstructionNewLocal2018} since they have good expansion properties and their links are flag complexes. It is interesting whether we get good qLDPC codes if the local codes $\cF_\sigma$ are the codes discussed in Section~\ref{ssc:flag-code}. It is natural to expect that, as in the case of expander LP codes, we get linear distance only if the local codes $\cF_\sigma$, $\sigma\in X(0)$, are good two-way expanders.

In the current paper, we focused on sheaves of classical codes. However, it is also possible to consider sheaves $\cQ$ of quantum CSS codes, where we assign to every open set $U\subseteq X$ the CSS code $\cQ(U)$ instead of a~classical linear code. It is interesting whether one can use this idea to analyze qLDPC codes.

\appendices

\enlargethispage{-1.2cm} 

\bibliographystyle{IEEEtran}
\bibliography{codes.bib}

\begin{thebibliography}{10}
\providecommand{\url}[1]{#1}
\csname url@samestyle\endcsname
\providecommand{\newblock}{\relax}
\providecommand{\bibinfo}[2]{#2}
\providecommand{\BIBentrySTDinterwordspacing}{\spaceskip=0pt\relax}
\providecommand{\BIBentryALTinterwordstretchfactor}{4}
\providecommand{\BIBentryALTinterwordspacing}{\spaceskip=\fontdimen2\font plus
\BIBentryALTinterwordstretchfactor\fontdimen3\font minus \fontdimen4\font\relax}
\providecommand{\BIBforeignlanguage}[2]{{%
\expandafter\ifx\csname l@#1\endcsname\relax
\typeout{** WARNING: IEEEtran.bst: No hyphenation pattern has been}%
\typeout{** loaded for the language `#1'. Using the pattern for}%
\typeout{** the default language instead.}%
\else
\language=\csname l@#1\endcsname
\fi
#2}}
\providecommand{\BIBdecl}{\relax}
\BIBdecl

\bibitem{Dinur:2021}
I.~Dinur, S.~Evra, R.~Livne, A.~Lubotzky, and S.~Mozes, ``Locally testable codes with constant rate, distance, and locality,'' Nov. 2021.

\bibitem{Panteleev&Kalachev:2021:ltc}
\BIBentryALTinterwordspacing
P.~Panteleev and G.~Kalachev, ``Asymptotically good quantum and locally testable classical ldpc codes,'' Nov. 2021. [Online]. Available: \url{http://arxiv.org/abs/2111.03654}
\BIBentrySTDinterwordspacing

\bibitem{leverrierQuantumTannerCodes2022a}
A.~Leverrier and G.~Z{\'e}mor, ``Quantum {{Tanner}} codes,'' in \emph{2022 {{IEEE}} 63rd {{Annual Symposium}} on {{Foundations}} of {{Computer Science}} ({{FOCS}})}.\hskip 1em plus 0.5em minus 0.4em\relax {IEEE Computer Society}, Oct. 2022, pp. 872--883.

\bibitem{dinurGoodQuantumLDPC2023}
I.~Dinur, M.-H. Hsieh, T.-C. Lin, and T.~Vidick, ``Good {{Quantum LDPC Codes}} with {{Linear Time Decoders}},'' in \emph{Proceedings of the 55th {{Annual ACM Symposium}} on {{Theory}} of {{Computing}}}, ser. {{STOC}} 2023.\hskip 1em plus 0.5em minus 0.4em\relax {New York, NY, USA}: {Association for Computing Machinery}, Jun. 2023, pp. 905--918.

\bibitem{kalachevTwosidedRobustlyTestable2023}
G.~Kalachev and P.~Panteleev, ``Two-sided {{Robustly Testable Codes}},'' Aug. 2023.

\bibitem{golowichNLTSHamiltoniansStronglyExplicit2023}
L.~Golowich and T.~Kaufman, ``{{NLTS Hamiltonians}} and {{Strongly-Explicit SoS Lower Bounds}} from {{Low-Rate Quantum LDPC Codes}},'' Nov. 2023.

\bibitem{hopkinsExplicitLowerBounds2022a}
M.~Hopkins and T.-C. Lin, ``Explicit {{Lower Bounds Against $\Omega$}}(n)-{{Rounds}} of {{Sum-of-Squares}},'' in \emph{2022 {{IEEE}} 63rd {{Annual Symposium}} on {{Foundations}} of {{Computer Science}} ({{FOCS}})}, Oct. 2022, pp. 662--673.

\bibitem{guSingleshotDecodingGood2023}
S.~Gu, E.~Tang, L.~Caha, S.~H. Choe, Z.~He, and A.~Kubica, ``Single-shot decoding of good quantum {{LDPC}} codes,'' Jun. 2023.

\bibitem{Gu:stoc2023:qpdpc-decoder}
\BIBentryALTinterwordspacing
S.~Gu, C.~A. Pattison, and E.~Tang, ``An efficient decoder for a linear distance quantum {LDPC} code,'' in \emph{Proceedings of the 55th Annual ACM Symposium on Theory of Computing}, ser. STOC 2023.\hskip 1em plus 0.5em minus 0.4em\relax New York, NY, USA: Association for Computing Machinery, 2023, p. 919–932. [Online]. Available: \url{https://doi.org/10.1145/3564246.3585169}
\BIBentrySTDinterwordspacing

\bibitem{Leverrier:qldpcdecoder:2023}
\BIBentryALTinterwordspacing
A.~Leverrier and G.~Zémor, \emph{Efficient decoding up to a constant fraction of the code length for asymptotically good quantum codes}, 2023, pp. 1216--1244. [Online]. Available: \url{https://epubs.siam.org/doi/abs/10.1137/1.9781611977554.ch45}
\BIBentrySTDinterwordspacing

\bibitem{anshuNLTSHamiltoniansGood2023}
A.~Anshu, N.~P. Breuckmann, and C.~Nirkhe, ``{{NLTS Hamiltonians}} from {{Good Quantum Codes}},'' in \emph{Proceedings of the 55th {{Annual ACM Symposium}} on {{Theory}} of {{Computing}}}, ser. {{STOC}} 2023.\hskip 1em plus 0.5em minus 0.4em\relax {New York, NY, USA}: {Association for Computing Machinery}, Jun. 2023, pp. 1090--1096.

\bibitem{williamsonLayerCodes2023}
D.~J. Williamson and N.~Baspin, ``Layer {{Codes}},'' Sep. 2023.

\bibitem{portnoyLocalQuantumCodes2023}
E.~Portnoy, ``Local {{Quantum Codes}} from {{Subdivided Manifolds}},'' Mar. 2023.

\bibitem{linGeometricallyLocalQuantum2023}
T.-C. Lin, A.~Wills, and M.-H. Hsieh, ``Geometrically {{Local Quantum}} and {{Classical Codes}} from {{Subdivision}},'' Sep. 2023.

\bibitem{dinurHighDimensionalExpanders2017}
I.~Dinur and T.~Kaufman, ``High {{Dimensional Expanders Imply Agreement Expanders}},'' in \emph{2017 {{IEEE}} 58th {{Annual Symposium}} on {{Foundations}} of {{Computer Science}} ({{FOCS}})}, Oct. 2017, pp. 974--985.

\bibitem{Meshulam:2018}
R.~Meshulam, ``Graph codes and local systems,'' Mar. 2018.

\bibitem{Breuckmann:balanced:2021}
N.~P. Breuckmann and J.~N. Eberhardt, ``Balanced product quantum codes,'' \emph{IEEE Transactions on Information Theory}, vol.~67, no.~10, pp. 6653--6674, Oct. 2021.

\bibitem{firstGoodQueryLocally2023}
U.~A. First and T.~Kaufman, ``On {{Good}} $2$-{{Query Locally Testable Codes}} from {{Sheaves}} on {{High Dimensional Expanders}},'' May 2023.

\bibitem{dinurNewCodesHigh2023}
I.~Dinur, S.~Liu, and R.~Y. Zhang, ``New {{Codes}} on {{High Dimensional Expanders}},'' Aug. 2023.

\bibitem{dinurqLTC:2024}
I.~Dinur, T.-C. Lin, and T.~Vidick, ``Expansion of higher-dimensional cubical complexes with application to quantum locally testable codes,'' Feb. 2024.

\bibitem{Kitaev:2002}
E.~Dennis, A.~Kitaev, A.~Landahl, and J.~Preskill, ``Topological quantum memory,'' \emph{Journal of Mathematical Physics}, vol.~43, no.~9, pp. 4452--4505, 2002.

\bibitem{Panteleev&Kalachev:2021}
P.~Panteleev and G.~Kalachev, ``Quantum {LDPC} codes with almost linear minimum distance,'' \emph{IEEE Transactions on Information Theory}, vol.~68, no.~1, pp. 213--229, Jan. 2022.

\bibitem{Hastings:2021:fiber}
M.~B. Hastings, J.~Haah, and R.~O'Donnell, ``Fiber bundle codes: breaking the {$N^{1/2} \operatorname{polylog}(N)$} barrier for quantum {LDPC} codes,'' in \emph{Proceedings of the 53rd Annual ACM SIGACT Symposium on Theory of Computing}.\hskip 1em plus 0.5em minus 0.4em\relax New York, NY, USA: Association for Computing Machinery, Jun. 2021, pp. 1276--1288.

\bibitem{Zemor:2001}
G.~{Z\'{e}mor}, ``On expander codes,'' \emph{IEEE Transactions on Information Theory}, vol.~47, no.~2, pp. 835--837, 2001.

\bibitem{dinurLocallyTestableCodes2022}
I.~Dinur, S.~Evra, R.~Livne, A.~Lubotzky, and S.~Mozes, ``Locally testable codes with constant rate, distance, and locality,'' in \emph{Proceedings of the 54th {{Annual ACM SIGACT Symposium}} on {{Theory}} of {{Computing}}}, ser. {{STOC}} 2022.\hskip 1em plus 0.5em minus 0.4em\relax {New York, NY, USA}: {Association for Computing Machinery}, Jun. 2022, pp. 357--374.

\bibitem{chenMaximallyRecoverableProperty2007}
M.~Chen, C.~Huang, and J.~Li, ``On the {{Maximally Recoverable Property}} for {{Multi-Protection Group Codes}},'' in \emph{2007 {{IEEE International Symposium}} on {{Information Theory}}}, Jun. 2007, pp. 486--490.

\bibitem{gopalanExplicitMaximallyRecoverable2014}
P.~Gopalan, C.~Huang, B.~Jenkins, and S.~Yekhanin, ``Explicit {{Maximally Recoverable Codes With Locality}},'' \emph{IEEE Transactions on Information Theory}, vol.~60, no.~9, pp. 5245--5256, Sep. 2014.

\bibitem{tamoFamilyOptimalLocally2014}
I.~Tamo and A.~Barg, ``A {{Family}} of {{Optimal Locally Recoverable Codes}},'' \emph{IEEE Transactions on Information Theory}, vol.~60, no.~8, pp. 4661--4676, Aug. 2014.

\bibitem{golowichQuantumLocallyRecoverable2023}
L.~Golowich and V.~Guruswami, ``Quantum {{Locally Recoverable Codes}},'' Nov. 2023.

\bibitem{kalachev:2024}
G.~Kalachev and P.~Panteleev, ``Maximally extendable product codes are good coboundary expanders,'' 2024, in preparation.

\bibitem{Aharonov:2015}
D.~Aharonov and L.~Eldar, ``Quantum locally testable codes,'' \emph{SIAM Journal on Computing}, vol.~44, no.~5, pp. 1230--1262, Jan. 2015.

\bibitem{eldarLocalHamiltoniansWhose2017}
L.~Eldar and A.~W. Harrow, ``Local {{Hamiltonians Whose Ground States Are Hard}} to {{Approximate}},'' in \emph{2017 {{IEEE}} 58th {{Annual Symposium}} on {{Foundations}} of {{Computer Science}} ({{FOCS}})}.\hskip 1em plus 0.5em minus 0.4em\relax {Berkeley, CA, USA}: {IEEE}, Oct. 2017, pp. 427--438.

\bibitem{crossQuantumLocallyTestable2023}
A.~Cross, Z.~He, A.~Natarajan, M.~Szegedy, and G.~Zhu, ``Quantum {{Locally Testable Code}} with {{Constant Soundness}},'' Jul. 2023.

\bibitem{rungtanapiromInfiniteSeriesQuaternionic2019}
N.~Rungtanapirom, J.~Stix, and A.~Vdovina, ``Infinite series of quaternionic 1-vertex cube complexes, the doubling construction, and explicit cubical {{Ramanujan}} complexes,'' \emph{International Journal of Algebra and Computation}, vol.~29, no.~06, pp. 951--1007, Sep. 2019.

\bibitem{CSS:1996}
\BIBentryALTinterwordspacing
A.~R. Calderbank and P.~W. Shor, ``Good quantum error-correcting codes exist,'' \emph{Phys. Rev. A}, vol.~54, pp. 1098--1105, Aug 1996. [Online]. Available: \url{https://link.aps.org/doi/10.1103/PhysRevA.54.1098}
\BIBentrySTDinterwordspacing

\bibitem{CSS2:1996}
\BIBentryALTinterwordspacing
A.~M. Steane, ``Error correcting codes in quantum theory,'' \emph{Phys. Rev. Lett.}, vol.~77, pp. 793--797, Jul 1996. [Online]. Available: \url{https://link.aps.org/doi/10.1103/PhysRevLett.77.793}
\BIBentrySTDinterwordspacing

\bibitem{assmusCategoryLinearCodes1998}
E.~Assmus, ``The category of linear codes,'' \emph{IEEE Transactions on Information Theory}, vol.~44, no.~2, pp. 612--629, Mar. 1998.

\bibitem{Macwilliams:1962}
F.~J. MacWilliams, ``Combinatorial problems of elementary {{Abelian}} groups,'' Ph.D. dissertation, {Harvard University}, 1962.

\bibitem{Panteleev:2022:qmath}
P.~Panteleev, ``Sheaf codes,'' \href{https://www.math.ucdavis.edu/~qmath/assets/files/Schedule_and_abstracts_print.pdf}{Schedule and abstracts}, Sep. 2022, qMATH 15.

\bibitem{Tanner:1981}
R.~Tanner, ``A recursive approach to low complexity codes,'' \emph{IEEE Transactions on Information Theory}, vol.~27, no.~5, pp. 533--547, Sep. 1981.

\bibitem{Sipser:1996}
M.~Sipser and D.~Spielman, ``Expander codes,'' \emph{IEEE Transactions on Information Theory}, vol.~42, no.~6, pp. 1710--1722, Nov. 1996.

\bibitem{LiftingSmallLocally2019}
P.~Harsha, ``Lifting small locally testable codes ({{LTCs}}) to large {{LTCs}} via {{HDXs}}, {{Institute}} for {{Advanced Study}},'' \href{https://www.ias.edu/video/csdm/2019/1125-PrahladhHarsha}{Video}, Nov. 2019.

\bibitem{guoNewAffineinvariantCodes2013}
A.~Guo, S.~Kopparty, and M.~Sudan, ``New affine-invariant codes from lifting,'' in \emph{Proceedings of the 4th Conference on {{Innovations}} in {{Theoretical Computer Science}}}, ser. {{ITCS}} '13.\hskip 1em plus 0.5em minus 0.4em\relax {New York, NY, USA}: {Association for Computing Machinery}, Jan. 2013, pp. 529--540.

\bibitem{lavauzelleLiftedProjectiveReed2019}
J.~Lavauzelle, ``Lifted projective {{Reed}}{\textendash}{{Solomon}} codes,'' \emph{Designs, Codes and Cryptography}, vol.~87, no.~7, pp. 1541--1575, Jul. 2019.

\bibitem{brakensiekGenericReedSolomonCodes2023}
J.~Brakensiek, S.~Gopi, and V.~Makam, ``Generic {{Reed-Solomon Codes Achieve List-Decoding Capacity}},'' in \emph{Proceedings of the 55th {{Annual ACM Symposium}} on {{Theory}} of {{Computing}}}, ser. {{STOC}} 2023.\hskip 1em plus 0.5em minus 0.4em\relax {New York, NY, USA}: {Association for Computing Machinery}, Jun. 2023, pp. 1488--1501.

\bibitem{motwaniRandomizedAlgorithms1995}
R.~Motwani and P.~Raghavan, \emph{Randomized {{Algorithms}}}.\hskip 1em plus 0.5em minus 0.4em\relax {Cambridge University Press}, Aug. 1995.

\bibitem{rodierCodesFlagVarieties2003}
F.~Rodier, ``Codes from flag varieties over a finite field,'' \emph{Journal of Pure and Applied Algebra}, vol. 178, no.~2, pp. 203--214, Mar. 2003.

\bibitem{dinurGoodLocallyTestable2022}
I.~Dinur, S.~Evra, R.~Livne, A.~Lubotzky, and S.~Mozes, ``Good {{Locally Testable Codes}},'' Jul. 2022.

\bibitem{hansenSpectralTheoryCellular2019}
J.~Hansen and R.~Ghrist, ``Toward a spectral theory of cellular sheaves,'' \emph{Journal of Applied and Computational Topology}, vol.~3, no.~4, pp. 315--358, Dec. 2019.

\bibitem{currySheavesCosheavesApplications2014}
J.~M. Curry, ``Sheaves, cosheaves and applications,'' Ph.D. dissertation, University of Pennsylvania, 2014.

\bibitem{bjorner1984posets}
A.~Bj{\"o}rner, ``Posets, regular cw complexes and bruhat order,'' \emph{European Journal of Combinatorics}, vol.~5, no.~1, pp. 7--16, 1984.

\bibitem{Kaufman:2014}
T.~Kaufman, D.~Kazhdan, and A.~Lubotzky, ``Ramanujan complexes and bounded degree topological expanders,'' in \emph{2014 IEEE 55th Annual Symposium on Foundations of Computer Science}.\hskip 1em plus 0.5em minus 0.4em\relax Philadelphia, PA, USA: IEEE, Oct. 2014, pp. 484--493.

\bibitem{Lubotzky:2018}
\BIBentryALTinterwordspacing
A.~Lubotzky, ``High dimensional expanders,'' in \emph{Proceedings of the International Congress of Mathematicians (ICM 2018)}.\hskip 1em plus 0.5em minus 0.4em\relax Rio de Janeiro, Brazil: WORLD SCIENTIFIC, Jun. 2018, pp. 705--730. [Online]. Available: \url{https://www.worldscientific.com/doi/abs/10.1142/9789813272880_0027}
\BIBentrySTDinterwordspacing

\bibitem{Ben-Sasson:2006}
\BIBentryALTinterwordspacing
E.~{Ben-Sasson} and M.~Sudan, ``Robust locally testable codes and products of codes,'' \emph{Random Structures \& Algorithms}, vol.~28, no.~4, pp. 387--402, 2006. [Online]. Available: \url{https://onlinelibrary.wiley.com/doi/abs/10.1002/rsa.20120}
\BIBentrySTDinterwordspacing

\bibitem{Kalachev:example:2023}
\BIBentryALTinterwordspacing
G.~Kalachev, ``High-dimensional expansion of product codes is stronger than robust and agreement testability,'' Aug. 2023. [Online]. Available: \url{https://arxiv.org/abs/2308.02889}
\BIBentrySTDinterwordspacing

\bibitem{NetSheaf:2011}
R.~Ghrist and Y.~Hiraoka, ``Network codings and sheaf cohomology,'' in \emph{International Symposium on Nonlinear Theory and its Applications (NOLTA2011)}.\hskip 1em plus 0.5em minus 0.4em\relax Kobe, Japan: IEICE, Sep. 2011, pp. 266--269.

\bibitem{color2006}
H.~Bombin and M.~A. {Martin-Delgado}, ``Topological {{Quantum Distillation}},'' \emph{Physical Review Letters}, vol.~97, no.~18, p. 180501, Oct. 2006.

\bibitem{kaufmanConstructionNewLocal2018}
T.~Kaufman and I.~Oppenheim, ``Construction of new local spectral high dimensional expanders,'' in \emph{Proceedings of the 50th {{Annual ACM SIGACT Symposium}} on {{Theory}} of {{Computing}}}, ser. {{STOC}} 2018.\hskip 1em plus 0.5em minus 0.4em\relax {New York, NY, USA}: {Association for Computing Machinery}, Jun. 2018, pp. 773--786.

\bibitem{Tillich&Zemor:2009}
J.~Tillich and G.~Z{\'e}mor, ``Quantum {LDPC} codes with positive rate and minimum distance proportional to {{n{\textsuperscript{1/2}}}},'' in \emph{2009 IEEE international symposium on information theory}.\hskip 1em plus 0.5em minus 0.4em\relax Seoul, Korea (South): IEEE, Jun. 2009, pp. 799--803.

\bibitem{Kovalev:2013}
A.~A. Kovalev and L.~P. Pryadko, ``Quantum kronecker sum-product low-density parity-check codes with finite rate,'' \emph{Physical Review A}, vol.~88, no.~1, p. 012311, Jul. 2013.

\bibitem{linQuantumTwoblockGroup2023}
H.-K. Lin and L.~P. Pryadko, ``Quantum two-block group algebra codes,'' Jun. 2023.

\bibitem{bravyiHighthresholdLowoverheadFaulttolerant2023}
S.~Bravyi, A.~W. Cross, J.~M. Gambetta, D.~Maslov, P.~Rall, and T.~J. Yoder, ``High-threshold and low-overhead fault-tolerant quantum memory,'' Aug. 2023.

\bibitem{xuConstantOverheadFaultTolerantQuantum2023}
Q.~Xu, J.~P.~B. Ataides, C.~A. Pattison, N.~Raveendran, D.~Bluvstein, J.~Wurtz, B.~Vasic, M.~D. Lukin, L.~Jiang, and H.~Zhou, ``Constant-{{Overhead Fault-Tolerant Quantum Computation}} with {{Reconfigurable Atom Arrays}},'' Aug. 2023.

\bibitem{Panteleev&Kalachev:2019}
\BIBentryALTinterwordspacing
P.~Panteleev and G.~Kalachev, ``Degenerate quantum {LDPC} codes with good finite length performance,'' \emph{Quantum}, vol.~5, p. 585, Nov. 2021. [Online]. Available: \url{https://quantum-journal.org/papers/q-2021-11-22-585/}
\BIBentrySTDinterwordspacing

\bibitem{yangSpatiallyCoupledQDLPCCodes2023}
S.~Yang and R.~Calderbank, ``Spatially-{{Coupled QDLPC Codes}},'' Sep. 2023.

\bibitem{borelloDihedralQuantumCodes2023}
M.~Borello, A.-L. Horlemann, H.~Islam, and N.~Willenborg, ``Dihedral {{Quantum Codes}},'' Oct. 2023.

\bibitem{Brown:1982}
K.~S. Brown, \emph{Cohomology of {{Groups}}}, ser. Graduate {{Texts}} in {{Mathematics}}.\hskip 1em plus 0.5em minus 0.4em\relax {New York, NY}: {Springer}, 1982, vol.~87.

\bibitem{barmakAlgebraicTopologyFinite2011}
J.~A. Barmak, \emph{Algebraic Topology of Finite Topological Spaces and Applications}.\hskip 1em plus 0.5em minus 0.4em\relax {Springer}, 2011, vol. 2032.

\bibitem{margolisCellComplexesPoset2021}
S.~Margolis, F.~Saliola, and B.~Steinberg, \emph{Cell Complexes, Poset Topology and the Representation Theory of Algebras Arising in Algebraic Combinatorics and Discrete Geometry}.\hskip 1em plus 0.5em minus 0.4em\relax {American Mathematical Society}, 2021, vol. 274.

\end{thebibliography}

\section{General Sheaf Codes}\label{sc:gen-sheaf}

\subsection{Alexandrov Topology and Cellular Sheaves}

The definitions from Section~\ref{ssc:topology} are very abstract, and they do not give any insight into how one can construct a~sheaf of codes on a~topological space. For this reason, we limit ourselves to \emph{Alexandrov spaces}, i.e., the topological spaces where the intersection of \emph{any} collection of open sets is always open. This is not restrictive since we are interested in finite topological spaces, and they are obviously Alexandrov spaces. If we further require that a~space satisfies the Kolmogorov $T_0$-axiom\footnote{A topological space $X$ is a~\emph{$T_0$-space} if the points can be \emph{distinguished} by open sets, i.e., for every two different points $x,y\in X$ there exists an open set $U$ such that either $x\in U$, $y\not\in U$  or $y\in U$, $x\not\in U$. This condition is also not very restrictive since every topological space $X$ can be converted into $T_0$-space as the quotient $X/\sim$ by the indistinguishability of points.}, then any Alexandrov space can be viewed as a~poset\footnote{Without $T_0$-axiom every Alexandrov space $X$ can be identified with \emph{preordered sets} if we define the \emph{preorder} as $a \le b$ \Iff $[\{a\}] \subseteq [\{b\}]$, where $[\cdot]$ is the closure in $X$. If $X$ is a~$T_0$-space, then this preordered set is a~poset.} $X$, where the topology is given by the \emph{upper sets}, i.e., a~set $U\subseteq X$ is open \Iff $x\in U$ and $x\le y$ implies $y\in U$. 

It is clear that the open sets $X_{\ge \sigma } = \{x \in X\mid x \ge \sigma \}$ (the smallest open neighborhood of $\sigma\in X$) form the base of this topology. 
The main idea of \emph{cellular sheaves} is to assign the sets of local sections $\cF_{\sigma} = \cF(X_{\ge\sigma})$ called \emph{stalks} only to the open sets $X_{\ge\sigma}$ from this base, and then extend $\cF(U)$ to \emph{all} open sets $U = \cup_{\sigma\in U} X_{\ge\sigma}$.  Topologically, stalks $\cF_\sigma$ capture local properties\footnote{Every sheaf $\cF$ on a~space $X$ can also be considered as a~topological space $E_\cF$ called its \emph{{\'e}tale space},  a~generalization of a~covering of $X$ with the stalks $\cF_x$ playing the role of fibers. The elements of $E_\cF$ are all possible local sections $s\in \cF_x$ from stalks $\cF_x$, $x\in X$, and the open sets of $E_\cF$ are all compatible sets $\{s_i\}_{i\in I}$ of such local sections. We can also define a~map $p\colon E_\cF\to X$ (a~generalization of a covering projection), assigning $s\in \cF_x$ the element $x\in X$. It is not hard to show that the \emph{local sections} of $p$, i.e., the functions $s\colon U\to E_\cF$ such that $p\circ s = \id_U$, are in one-to-one correspondence with the elements of $\cF(U)$, which explains why these elements are called ``local sections''.} of the sheaf $\cF$ at points $\sigma\in X$. The adjective ``cellular'' here suggests that our model example is a~cell complex $X$ such as a~graph, a~simplicial, or a~cubical complex. 

These observations motivate the following definition~(see~\cite{hansenSpectralTheoryCellular2019} for further details). 

\begin{definition}[cellular sheaf, local system]\label{df:cellular-sheaf}
    Given a~poset $X$ and a category $\cC$ (e.g., $\Ab$, $\Rng$, $\Vc_\F$, or $\LC_\F$), a~\emph{cellular sheaf} or \emph{local system} on $X$ with values in $\cC$ is a~function $\cF$ assigning to every $\sigma$ from $X$ an object $\cF_\sigma$ from $\cC$ and to every pair $\sigma\le \tau$  a~morphism $\cF_{\sigma\to \tau}\colon \cF_\sigma\to\cF_\tau$ from $\cC$ such that for every triple $\sigma\le \tau\le \pi$ we have that: 
    \[\cF_{\tau\to\pi} \circ \cF_{\sigma\to\tau} = \cF_{\sigma\to\pi} \text{ and } \cF_{\sigma\to \sigma} = \id_{\cF_\sigma}.\]
\end{definition}

Now, if the category $\cC$ is sufficiently nice\footnote{For example, it has all limits (see \cite{currySheavesCosheavesApplications2014}). Unfortunately, the category $\LC_\F$ does not have this property.}, we can easily extend $\cF(U)$ to all open sets $U = \cup_{\sigma\in U} X_{\ge\sigma}$ if define $\cF(U)$ to be the set of all \emph{compatible collections} $(s_\sigma)_{\sigma\in U}$ of local sections, i.e., where $\cF_{\sigma\to \tau}(s_\sigma) = s_\tau$ for every pair $\sigma \le \tau$ from $U$. If $\cC$ is $\Ab$ (resp., $\Rng$, $\Vc_\F$), the sets of compatible collections $(s_\sigma)_{\sigma\in U}$ can easily be shown to have a~structure of abelian group (resp., ring, vector space), which makes $\cF(U)$ an object from $\cC$. 

In the case of codes, it is convenient to assume that all codes $\cF_\sigma$ have their index sets $X_\sigma$ inside some fixed set $X_*$, and all the maps $\cF_{\sigma\to\tau}$ are the standard restriction maps $x\mapsto x|_{X_\sigma}$. This leads us to the definition of sheaf codes obtained from Tanner codes, given in Section~\ref{sc:tanner-sheaf} (see also \cite{dinurNewCodesHigh2023, dinurqLTC:2024}), which we use as our main definition of a~sheaf code throughout the paper. 

However, in this section, we give a~more general definition of a~sheaf code that works with arbitrary cellular sheaves of linear codes. The main problem, in the case of codes (i.e., $\cC$ is $\LC_\F$), that it is not a~good idea to define $\cF(U)$ just as the set $\tilde{\cF}(U)$ of concatenations $(s_\sigma)_{\sigma\in U}$ of pairwise compatible codewords $s_\sigma\in \cF_\sigma \le \F^{X_\sigma}$ since the imposed compatible constrains $\cF_{\sigma\to \tau}(s_\sigma) = s_\tau$ may imply that some symbols $s_\tau(a)$ and $s_\sigma(b)$ in \emph{different} codewords are \emph{always} proportional, i.e., $s_\tau(a) = \gamma s_\sigma(b)$ for some fixed $\gamma\in\F^{\times}$. 

This happens when in the matrix\footnote{For linear operators $\cL\colon \F^A\to\F^B$, here we denote its matrix $\bcL\in \F^{A\times B}$ using bold fonts to avoid confusion.} $\bcF_{\sigma\to\tau}\in \F^{X_\tau \times X_\sigma }$ we have $\bcF_{\sigma\to\tau} (a, b) = \gamma$, in which case we write $a \sim_U b$. Now, considering the equivalence relation generated\footnote{Here we mean the reflexive and transitive closure of the relation $\sim_U$, which we also denote by $\sim_U$.} by $\sim_U$ on $I = \sqcup_{\sigma\in U} X_\sigma$, we can choose a~set of representatives $X_U \subseteq I$ from each equivalence class and put 
\[\cF(U) = \{c|_{X_U} : c\in \tilde{\cF}(U)\},\] 
where the set $\hat{\cF}(U)$ of pairwise compatible collections of codewords $(s_\sigma)_{\sigma\in U}\in \F^I$ is the~linear code defined as
\[
\tilde{\cF}(U) = \{c \in\F^{I} : c|_{X_\sigma}\in\cF_\sigma, \sigma\in U, R_U c = 0\}.
\]
Here $R_U$ is the matrix of the set of homogeneous linear equations: $\cF_{\sigma\to \tau}(s_\sigma) - s_\tau$ = 0, $\sigma\le \tau$; $\sigma,\tau\in U$. Since we get $\cF(U)$ as a~projection of the~larger linear code $\tilde{\cF}(U)$ on its subset of indices $X_U \subseteq I$, then $\cF(U)$ is also a linear code. Moreover, since $a \sim_U b$ implies that $s_\tau(a)=\gamma_{ab} s_\sigma(b)$ in $\tilde{\cF}(U)$ for some fixed $\gamma_{ab}\in\F^\times$, we see that different sets of representatives $X_U \subseteq I$ lead to monomially equivalent codes, and hence isomorphic in the category $\LC_\F$. It is also not hard to check that the sets $X_U \subseteq I$ are exactly the information sets of the linear code $\ker R_U$.

\section{Product Codes}\label{sc:lp-codes}

In this section, we consider tensor product codes, and their generalizations called the \emph{hypergraph product} and \emph{lifted product} codes.

\subsection{Classical Product Codes and Quantum Hypergraph Product Codes}

We start from the natural relation between the classical
tensor product codes and the quantum hypergraph product codes. \

Given two linear codes $\mathcal{A} \leqslant \mathbb{F}^{n_a}$ and
$\mathcal{B} \leqslant \mathbb{F}^{n_b}$, the ({\tmem{tensor}})
{\tmem{product}} code is the linear code defined as
\[ \mathcal{A} \otimes \mathcal{B} \!=\! \{ w \!\in\! \mathbb{F}^{n_a \times
   n_b} \mid \forall i,j : w (\cdot, j) \in
   \mathcal{A}, w (i, \cdot) \in \mathcal{B} \}, \]
consisting of all matrices where each column $w (\cdot, j)$ is from
$\mathcal{A}$ and each row $w (i, \cdot)$ is from $\mathcal{B}$.

If $A \in \mathbb{F}^{m_a \times n_a}$ and $B \in \mathbb{F}^{m_b \times n_b}$
are respectively some parity-check matrices\footnote{We
do not assume here that those matrices have full rank. } for the codes
$\mathcal{A}$ and $\mathcal{B}$, then since $\mathcal{A} \otimes
\mathbb{F}^{n_b} = \{ w \in \mathbb{F}^{n_a \times n_b} \mid Aw = 0 \}$ and
${\mathbb{F}^{n_a} \otimes \mathcal{B}}= \{ w \in \mathbb{F}^{n_a \times n_b}
\mid wB^{\Tau} = 0 \}$, we can also give this definition in a~more concise
form:
\begin{equation*}
  \mathcal{A} \otimes \mathcal{B} \assign \{ w \in \mathbb{F}^{n_a \times n_b}
  \mid Aw = 0, wB^{\Tau} = 0 \}, \label{eq:prod}
\end{equation*}
where we check that the {\tmem{vertical}} (resp., {\tmem{horizontal}})
syndromes collected in the matrix $v \assign Aw$ (resp., $u = wB^{\Tau}$)
are all equal to zero. Note that those syndromes have a lot of dependencies
since from $(Aw) B^{\Tau} = A (wB^{\Tau})$ immediately follows the equation
$Au = vB^{\Tau}$.

Using these observations, we can define the quantum CSS code $\tmop{HP} (A,
B) = \cC_X/\cC^\perp_Z$ called the {\tmem{hypergraph product}} ({\tmem{HP}}) code, where the
codewords from $\cC_X$ are all the pairs ${(u, v)\in \mathbb{F}^{n_a \times
  n_b} \times \mathbb{F}^{m_a \times m_b}}$, satisfying the equation $Au = vB^{\Tau}$, and the
\emph{degenerate} ones (i.e., from $\cC_Z^\perp$) are all possible syndromes $(wB^{\Tau}, Aw)$. In~the case of
hypergraph product codes, it is usually more convenient to consider the codes
$\tmop{HP} (A, B^{\Tau})$ instead of $\tmop{HP} (A, B)$, in which case we have
\begin{equation*}
  \tmop{HP} (A, B^{\Tau})\! \assign\! \frac{\{ (u, v)\! \in\! \mathbb{F}^{n_a \times
  n_b} \!\times\! \mathbb{F}^{m_a \times m_b} \mid Au = vB \}}{\{ (wB, Aw) \mid w
  \in \mathbb{F}^{n_a \times m_b} \}} . \label{eq:HP}
\end{equation*}
It is known~{\cite{Tillich&Zemor:2009}} that given a linear $[n, k, d]$-code
with a full rank parity-check matrix $H \in \mathbb{F}^{m \times n}$, the code
$\tmop{HP} (H, H^{\Tau})$ has the parameters $\llbracket n^2 + m^2, k^2, d
\rrbracket$.

\subsection{Lifted Product Codes}

Lifted product (LP) codes~\cite{Panteleev&Kalachev:2021} are a~natural generalization of classical product
codes and quantum hypergraph product codes~\cite{Tillich&Zemor:2009} discussed above. The key idea is
just to replace the field $\mathbb{F}$ in their definition by a~group
algebra\footnote{Note that LP codes can also be defined over
arbitrary algebra $R$, not necessary the group algebra, and they can also be
constructed from two chain complexes of arbitrary
length~{\cite{Panteleev&Kalachev:2021}} .} $R = \mathbb{F} [G]$, and then
consider codewords over $R$ as $| G |$ times larger codewords
over~$\mathbb{F}$ by representing each code symbol $\sum_{g \in G} a_g g \in
R$ by the block $(a_g)_{g \in G}$ of $| G |$ bits.

Given matrices ${A \in R^{m_a \times n_a}} $ and ${B \in R^{m_b
\times n_b}} $, we define the {\tmem{classical lifted product}} \ code as
\begin{equation}
  \mathcal{A} \otimes_R \mathcal{B} \assign \{ w \in R^{n_a \times n_b} \mid
  Aw = 0, wB^{\Tau} = 0 \},
\end{equation}
where the classical code $\mathcal{A} \otimes_R \mathcal{B}$ is uniquely
defined by the column code $\mathcal{A} \assign \{ w \in R^{n_a} \mid Aw = 0
\}$ and the row code $\mathcal{B} \assign \{ w \in R^{n_b} \mid wB^{\Tau} = 0
\}$, and it can also be described as
\[  \{ w \in R^{n_a \times n_b} \mid
   \forall i,j: w (\cdot, j) \in \mathcal{A}, w
   (i, \cdot) \in \mathcal{B} \} . \]

As in the case of hypergraph product codes, we define the {\tmem{quantum
lifted product}} code $\tmop{LP} (A, B)$ as the CSS code, where the codewords
are the pairs $(u, v)$ satisfying $Au = vB^{\Tau}$ and the degenerate ones are
all possible syndromes $(wB^{\Tau}, Aw)$. As usual, we are mostly interested
in the codes $\tmop{LP} (A, B^{\Tau})$ instead of $\tmop{LP} (A, B)$, in which
case they are defined as
\begin{equation}
   \frac{\{ (u, v) \in R^{n_a \times n_b}
  \times R^{m_a \times m_b} \mid Au = vB \}}{\{ (wB, Aw) \mid w \in R^{n_a
  \times m_b} \}}.
\end{equation}

Specializing this general definition to $1\times 1$ matrices, we get the \emph{generalized bicycle} (\emph{GB}) \emph{codes} from~\cite{Kovalev:2013}, also known as~\emph{2BGA codes}~\cite{linQuantumTwoblockGroup2023} in the case of arbitrary group~$G$:
\[ \tmop{LP} (a, b) \assign \frac{\{ (u, v) \in R \times R \mid au = vb \}}{\{
   (wb, aw) \mid w \in R \}}. \]

Such codes, as well as more general LP codes,  gain a~lot of attention recently since it was shown that they have similar performance to surface codes for the circuit level noise model~\cite{bravyiHighthresholdLowoverheadFaulttolerant2023, xuConstantOverheadFaultTolerantQuantum2023}, while having several times more logical qubits. Previously, this was only known for the much more simple depolarizing noise model~\cite{Panteleev&Kalachev:2019}.

Let us also briefly note that GB codes $\LP(a, b)$ have a~natural symmetry given by the~map $(u, v) \mapsto (\bar{v}, \bar{u})$, where $\overline{\sum_{i\in I} c_i\cdot g_i} = \sum_{i\in I} c_i\cdot g_i^{-1}$ is the \emph{antipode map}, that implies that $d_X(\LP(a, b)) = d_Z(\LP(a, b))$ and that we can do transversal Hadamard gates~\cite{Panteleev&Kalachev:2021, bravyiHighthresholdLowoverheadFaulttolerant2023}. 

Another interesting property of GB codes is that the Tanner graph of $\LP(a, b)$ coincides with the $1$-skeleton of the $4$-fold left-right Cayley complex $\cay_2(A, G, B)$, where $a = \sum_{x\in A}x$, $b=\sum_{x \in B} x$. These square complexes are at the heart of all currently known constructions of good qLDPC codes~\cite{Panteleev&Kalachev:2021:ltc, leverrierQuantumTannerCodes2022a, dinurGoodQuantumLDPC2023}. The codes in~\cite{Panteleev&Kalachev:2021:ltc, dinurGoodQuantumLDPC2023} are respectively the codes $\LP(Z_S(h_a), Z_S^\T(h_b))$ and $\LP(Z_S(h_a), Z_S(h_b))$, where 
\[ Z_S(h) =  \left(\begin{array}{ccc}
     1_G h^{(1)} & \ldots & 1_G h^{(\Delta)}\\
     s_1 h^{(1)} & \ldots & s_{\Delta} h^{(\Delta)}
   \end{array}\right)\in R^{2m\times\Delta} \]
is the parity-check matrix over $R=\F[G]$ of \emph{Z\'{e}mor code}~\cite{Zemor:2001}, a~variation of Sipser-Spielman code~\cite{Sipser:1996}, defined on a~double cover of a~Ramanujan Cayley graph $\cay(G;S)$, $S=\{s_i\}_{i\in [\Delta]}$, with the local code $\ker h$,  $h=(h^{(1)},\dots,h^{(\Delta)})\in\F^{m\times \Delta}$.

More examples of quasi-abelian and non-abelian LP codes can be found in~\cite{yangSpatiallyCoupledQDLPCCodes2023, borelloDihedralQuantumCodes2023}.

\section{Actions on Posets and Codes}\label{sc:action}
\subsection{$G$-sets and $G$-posets}
We say that a~group $G$ \emph{acts} on a~set $X$ and call $X$ a~\emph{$G$-set}, if we have two maps $x\mapsto g.x$ and $x\mapsto x.g$ named respectively the \emph{left} and the \emph{right} \emph{actions}\footnote{The \emph{composition} of two left (resp., right) actions is by definition the action $x\mapsto (g_2g_1).x$ (resp., $x\mapsto x.(g_1g_2)$.} of $G$ on $X$, such that $g.x = x.g^{-1}$ for all $x\in X$ and $g\in G$. If we do not explicitly mention what action we consider, it will always be the left action.
The \emph{quotient} of $X$ under the action of $G$, also known as the \emph{orbit space},  does not depend on which of two actions we choose, and is defined as $X/G = \{[x]_G \mid x\in X\}$, where $[x]_G = \{g.x \mid g\in G\} = \{x.g \mid g\in G\}$ is called the \emph{orbit} of~$x$. We also say that a~subset $S\subseteq X$ is \emph{$G$-invariant} if $g.S = \{g.x \mid x\in S\}=S$.

We say that a~poset $X$ is a~\emph{$G$-poset} if $G$ acts on $X$ (as a~set) and $a \le b$ implies $g.a \le g.b$ for all $g\in G$ (or equivalently $a.g \le b.g$). A~\emph{graded $G$-poset} is a~$G$-poset also respecting grading $\rho(g.x) = \rho(x)$.
Given a~graded $G$-poset $X$, we define its \emph{quotient} as the graded poset $X/G$ with the grading $\rho(A) = \rho(a)$ for $a\in A$ and the partial order defined as $A \le B$  \Iff there exist $a \in  A$ and $b\in B$ such that $a\le b$.  Let us show that $X/G$ is a~poset. The reflexivity $A\le A$ is trivial. If $A\le B$ and $B\le C$ then we get $a\le b$ and $b'\le c$ for $a\in A$, $b,b'\in B$, $c\in C$, and since there exists $g\in G$ such that $g.b = b'$ we have $g.a \le g.b \le c$, and hence $A\le C$. Finally, if $A < B$ and $B > A$, then we clearly have a~contradiction $\rho(A) < \rho(B)$ and $\rho(B) < \rho(A)$.

Let us recall that we consider graphs (resp. simplicial and cubical complexes) as graded posets $X$, and if $X$ is also a~graded $G$-poset, then we call it a~\emph{$G$-graph} (resp., simplicial or cubical \emph{$G$-complex}). 

\subsection{$G$-codes}
Given a~$G$-set $X$, we say that a~code $C\le \F^X$ is a~\emph{$G$-code} if $g.c = (c(g^{-1}.x))_{x\in X} \in \cC$ for every $g\in G$ (or equivalently $c.g = \rbr{c(g.x)}_{x\in X} \in \cC$), i.e., $G$-codes are the $G$-invariant subspaces $\cC\le \F^X$ for the group action $c\mapsto g.c$.  As~in the case of $G$-sets and $G$-posets, we can also define the~\emph{quotient} $\cC/G\le \F^{X/G}$ of a~$G$-code $\cC\le \F^X$ if we put 
\begin{equation}\label{eq:code-quo}
\cC/G=\{c \in \F^{X/G}\mid \rho_G(c)\in \cC\},    
\end{equation}
where the operator $\rho_G\colon c\mapsto \rbr{c([x]_G)}_{x\in X}$ \emph{repeats} $\abs{S}$ times each symbol $c(S)$, where $S = [x]_G$. Note that the above definition is very general, and it works even with \emph{nonlinear} codes. Moreover, the definition and the results from this section are still applicable even to codes $\cC\le\F^X$ where $X$ is a~$G$-set, but $\cC$ is \emph{not} $G$-invariant.    

From the above definition, it~is easy to see that we also have that $\cC/G=\{c \in \F^{X/G}\mid \rho_G(c)\in \cC^G\}$, where   
\[
\cC^G = \{c\in \cC \mid \forall g\in G\colon g.c = c\}
\]
is the code of \emph{$G$-invariants} used in~\cite{Breuckmann:balanced:2021}, implying that $\abs{\cC/G} = \abs{\cC^G}$ and $\dim{\cC/G} = \dim{\cC^G}$. In fact, another code from~\cite{Breuckmann:balanced:2021} called the code of \emph{$G$-coinvariants} 
\[
\cC_G = \{\sigma_G(c)  \in \F^{X/G} \mid c\in \cC \},
\]
where the operator $\sigma_G\colon c\mapsto \bigl(\sum_{x\in S} c(x)\bigr)_{S\in X/G}$ \emph{sums} the symbols in each orbit $S\in X/G$, is naturally related with the dual of the quotient $\cC/G$:
\begin{equation}\label{eq:coinv-code}
(\cC/G)^\perp = (\cC^\perp)_G,\quad (\cC_G)^\perp = \cC^\perp/G.
\end{equation}
Let us check the second equation $(\cC_G)^\perp = \cC^\perp/G$, since the first one immediately follows. Indeed, we have $a\in (\cC_G)^\perp$ \Iff $\forall b\in\cC\colon \dtp{a}{\sigma_G(b)} = \dtp{\rho_G(a)}{b} = 0$ \Iff $\rho_G(a)\in \cC^\perp$ \Iff $a\in \cC^\perp/G$,
where we used that
\[
\sum_{x\in X} a([x]_G)b(x)=\sum_{S\in X/G} a(S)\sum_{x\in S} b(x).
\]

Since $c\in \cC/G$ implies $\rho_G(c)\in \cC$, and we obtain $\abs{\rho_G(c)}\ge \abs{c}\min_{S\in X/G}\abs{S}$, then we have the following lower bound for the minimal distance 
\begin{equation}
d(\cC/G) \ge \frac{d(\cC)}{\min_{S\in X/G} \abs{S}}.     
\end{equation}

\begin{remark}
    In fact, we can also define the quotient $Q/G$ of a~quantum CSS code $Q=\cA/\cB$, where $\cA,\cB\le \F^{S}$ are both $G$-codes, if we put $Q/G = (\cA/G)/(\cB/G)$. 
\end{remark}

It is interesting that our quotient operation allows us to define the balanced product of two $G$-codes, which naturally corresponds to the  tensor product $\cA\otimes_R \cB$ of chain complexes $\cA,\cB$ over a~group ring $R = \F[G]$ used in \cite{Panteleev&Kalachev:2021}, \cite{Breuckmann:balanced:2021} to construct qLDPC codes. This general operation is often applied in group cohomology and is known to be related with the \emph{balanced product} of $G$-sets (e.g., see~\cite[Chapter~III]{Brown:1982}). In our definition below, we follow very closely the reference~\cite{Breuckmann:balanced:2021}, which considers the most general case of $\cA \otimes_R \cB$, but there are two important differences. First, the result of our operation is a~classical code\footnote{An~idea that we can also construct \emph{classical} codes using $\cA\otimes_R \cB$ was first proposed in~\cite{Panteleev&Kalachev:2021:ltc} to get good LTCs.}, not a~quantum code or chain complex, as it was in~\cite{Panteleev&Kalachev:2021}, \cite{Breuckmann:balanced:2021}. Second, our definition does not use parity-check matrices, and works for \emph{nonlinear} codes.

\begin{definition}[balanced product of codes]
    The \emph{balanced product} of $G$-codes $\cA\le \F^X$ and $\cB \le \F^Y$ is the code $\cA \otimes_G \cB\le \F^{X\times_G Y}$ defined as the quotient $\cA \otimes \cB/G$, where $G$ acts on $X\times Y$ as $g.(x,y) = (g.x, g.y)$. 
\end{definition}

This definition is correct since the diagonal map $(x,y)\mapsto (g.x, g.y)$ is the~composition $\Delta_R \circ \Delta_L$ of maps $\Delta_L\colon (x,y)\mapsto (g.x, y)$ and $\Delta_R\colon (x,y)\mapsto (x, g.y)$, and thus $\Delta_R(\Delta_L(\cA \otimes \cB)) \subseteq \Delta_R(\cA \otimes \cB) \subseteq \cA \otimes \cB$.

We can also define the quotients and the balanced products for \emph{sheaf $G$-codes}, i.e., sheaf codes $\cF$ on a~graded $G$-posets $X$ such that $c|_{X_\sigma}\in \cF_{\sigma}$ implies $(g.c)|_{X_{g.\sigma}} \in \cF_{g.\sigma}$ for $g\in G$. The last condition ensures that for every $G$-invariant $U\subseteq X$ we have that the code $\cF(U) = \{c\in \F^{X_U}\mid \forall \sigma\in U\colon c|_{X_\sigma}\in \cF_\sigma\} \le \F^{X_U}$ is a~$G$-code, and thus we can consider its quotient $\cF(U)/G$. This leads us to the following definition.
\begin{definition}[quotient sheaf code]
    For a~sheaf $G$-code $\cF$ on $X$ we define its \emph{quotient} $\cF/G$ as the sheaf code on $X/G$ with the global code $\cF/G(X/G) = \cF(X)/G$ and the local codes $(\cF/G)_S = \cF(S)/G$, $S\in X/G$. 
\end{definition}
The definition is correct since for $S\le T$ if we have $c\in \cF(S)/G$ then $c|_{(X/G)_T}\in \cF(T)/G$. In fact, one can check that we can also define this operation as the pushforward $\phi_* \cF$ (see Section~\ref{sc:oper-sheaf}), where $\phi\colon X\to X/G$ is a~natural \emph{projection} defined as $\phi\colon x\mapsto [x]_G$. 

Now, combining the operation of the product and the quotient for sheaf codes, we can also define their balanced product.
\begin{definition}[balanced product of sheaf codes]\label{df:balanced-prod-sheaf}
The \emph{balanced product} $\cA \otimes_G \cB$ of sheaf $G$-codes $\cA(X), \cB(Y)$ is the quotient sheaf code $\cA \otimes \cB/G$, where $G$ acts on $X\times Y$ as $g.(x,y) = (g.x, g.y)$.
\end{definition}
The coded space for $\cA \otimes_G \cB$ is the poset $X\times Y/G$, which leads us to the definition of balanced product for posets, an~instance of the~corresponding notion from topology if we view posets as Alexandrov topological spaces (see Appendix~\ref{sc:posets}).

\begin{definition}[balanced product of graded\footnote{In fact, this definition works with arbitrary posets of finite height. However, it is possible to give an example of posets $X$, $Y$ of \emph{infinite} height such that the quotient $X\times Y/G$ is only a~preordered set.} $G$-posets]\label{df:balanced-prod-poset}
    Given two graded $G$-posets $X$ and $Y$, their \emph{balanced product} is defined as the quotient $X\times Y/G$, where $G$ acts on $X\times Y$ as $g.(x,y) = (g.x, g.y)$. 
\end{definition}

If we apply this general operation to the special case when $X$ and $Y$ are both $G$-graphs~\cite{Breuckmann:balanced:2021}, then we can also define the \emph{balanced box product} of $G$-graphs $X\Box_G Y = \sk_{1} X\times_G Y$, which naturally generalizes the standard \emph{box product}\footnote{The box product is also known as the \emph{Cartesian product}.} of graphs $X\Box Y = \sk_1 X\times Y$.

\subsection{Free action}

Given a~$G$-set we say that the group $G$ acts \emph{freely} on $X$ or alternatively that the action of $G$ is \emph{free} if $g.x = x$ (or equivalently $x.g = x$) always implies $g=1_G$.  
A~$G$-set (resp., $G$-poset, $G$-graph, $G$-code) is called \emph{free} if the action of $G$ is free.
Free $G$-codes are known in coding theory as \emph{quasi-$G$} codes of \emph{index} $n=|X/G|$, and if $G$ is arbitrary (resp. abelian, cyclic) group, then the corresponding classes of codes are called \emph{quasi-group} (resp., \emph{quasi-abelian}, \emph{quasi-cyclic}) codes. It is common to consider the codewords of quasi-$G$ codes of index $n$ as elements from $R^n$, where $R = \F[G]$. Each such code $C$ can be represented either as\footnote{We represent codewords as column vectors, which is non-standard, but very convenient for our needs.} $\cC = \{c\in R^n \mid H_\lambda c = 0 \}$ or as $\cC = \{c\in R^n\mid c^\T H_\rho^\T = 0 \}$, where $H_\lambda$, $H_\rho$ are called its \emph{left} and \emph{right parity-check matrices}, respectively. Alternatively, we can represent them as $\cC = \{G^\T_\lambda u \mid u\in R^k \}$ or as $\cC = \{u^\T G_\rho \mid u\in R^k \}$, in which case we call $G_\lambda$ and $G_\rho$ the \emph{left} and the \emph{right} \emph{generator matrices}, respectively\footnote{We should warn the reader that though the left and the right versions of generator and parity-check matrices always exist for any quasi-$G$ code, they may be completely unrelated to each other and have even different sizes.}.  When $R=\F[G]$ is commutative, we do not need to distinguish the left and right versions, and can just put $H_\lambda = H_\rho = H$ and $G_\lambda = G_\rho = G$, representing the code as usual

\[
\cC = \{c\in R^n \mid Hc = 0\} = \{G^\T c \mid  c\in R^k\}. 
\]

Let us also note that since $(\cC/G)^\perp = (\cC^\perp)_G$, then if $H = (h_{ij})_{m\times n}$ is a~parity-check matrix (either left or right) of $\cC$, then $\eps(H) = (\eps(h_{ij}))_{m\times n}$ is the parity-check matrix of $\cC/G$, where $\eps\colon \F[G]\to \F$ is given by $\sum_{g\in G} a_g\cdot g\mapsto \sum_{g\in G} a_g$. Moreover, if $G$ is the cyclic group of order~$\ell$, then we have $\F[G]\cong \F[x]/(x^\ell-1)$, and $\eps(H)=H(1)$ is the parity-check matrix for $\cC/G$.

\section{Posets as Topological Spaces}\label{sc:posets}

Posets are very convenient for our purposes since they can naturally represent in a unified way~(e.g., see~\cite{Panteleev&Kalachev:2021:ltc}) various combinatorial objects like graphs, hypergraphs, simplicial, and cubical complexes, which often appear in the context of error correcting codes. Viewing them as topological spaces with Alexandrov topology~\cite{barmakAlgebraicTopologyFinite2011, margolisCellComplexesPoset2021} allows one to translate the standard topological definitions and constructions (e.g., continuous maps, homotopy, direct products, coverings, fiber bundles, balanced products, etc.) into the purely combinatorial language of posets.

For example, one can check that the continuous maps $f\colon X\to Y$ between posets $X$ and $Y$ (considered as topological spaces) are exactly the \emph{monotone maps}, i.e., $a\le b$ implies $f(a) \le f(b)$, which we also call \emph{morphisms} or \emph{poset maps}. Another example is the~\emph{direct product} of posets $X$ and $Y$ defined as the set $X \times Y$ with the partial order $(a,b) \le (a',b')$ \Iff $a\le a'$ and $b \le b'$, which corresponds to their direct product as topological spaces. One can also easily check that $G$-posets can be viewed as topological spaces with a~continuous action of $G$ (considered as a~discrete topological group). Thus, we can naturally apply the general idea of balanced product of topological spaces to $G$-posets as well (Definition~\ref{df:balanced-prod-poset}).

\subsection{Poset Codes}

In fact, it is even possible to assign to each poset $X$ a~simplicial complex called its \emph{order complex} $\Delta(X)$ with the set of vertices $X$ and all possible chains in $X$ as faces. A natural geometrical realization of $\Delta(X)$ in $\R^n$ can be viewed as a~geometrical realization of $X$, since the topological spaces $X$ and $\Delta(X)$ have essentially\footnote{In fact, they are \emph{weakly homotopy equivalent}, implying that their fundamental group and homology groups are the same.} the same properties. Moreover, if $X$ is the cell poset of some cell complex (e.g., simplicial or cubical), than $\Delta(X)$ is just its first barycentric subdivision. Since many existing constructions of LTCs and qLDPC codes can be obtained as (co)homology groups of complexes, this raises a~natural question:

\begin{problem}
Can we use (co)homology groups of $\Delta(X)$ to construct classical LTCs and quantum LDPC codes with good parameters?     
\end{problem}

Since every poset $X$ is defined by its Hasse diagram, which is just a~multipartite graph, potentially we have a~very general class of codes obtained from random sparse graphs.

\end{document}